\documentclass[review]{elsarticle}

\usepackage{lineno,hyperref}
\usepackage{amsmath}
\usepackage{amssymb}
\usepackage{algorithmic}
\usepackage{algorithm}
\usepackage{caption}
\usepackage{mathtools}
\usepackage{changes}
\usepackage{multirow}
\usepackage{verbatim}
\usepackage{mathrsfs}
\usepackage{graphicx}
\usepackage{amsmath,amsthm,bm,mathrsfs}

\theoremstyle{plain}
\newtheorem{theorem}{Theorem}[section]

\newtheorem{corollary}[theorem]{Corollary}
\newtheorem{proposition}[theorem]{Proposition}

\theoremstyle{definition}
\newtheorem{definition}[theorem]{Definition}

\theoremstyle{remark}
\newtheorem{remark}{Remark}

\modulolinenumbers[5]

\journal{ArXiv.org}

\begin{document}

\begin{frontmatter}

\title{Time-limited $\mathcal{H}_2$-optimal Model Order Reduction of Linear Systems with Quadratic Outputs}

\author[uz]{Umair~Zulfiqar\corref{mycorrespondingauthor}}
\cortext[mycorrespondingauthor]{Corresponding author}
\ead{umairzulfiqar@shu.edu.cn}
\author[zx]{Zhi-Hua~Xiao}
\author[qs]{Qiu-Yan~Song}
\author[mud]{Mohammad~Monir~Uddin}
\author[vs]{Victor~Sreeram}
\address[uz]{School of Electronic Information and Electrical Engineering, Yangtze University, Jingzhou, Hubei, 434023, China}
\address[zx]{School of Information and Mathematics, Yangtze University, Jingzhou, Hubei, 434023, China}
\address[qs]{School of Mechatronic Engineering and Automation, Shanghai University, Shanghai, 200444, China}
\address[mud]{Department of Mathematics and Physics, North South University, Dhaka, 1229, Bangladesh}
\address[vs]{Department of Electrical, Electronic, and Computer Engineering, The University of Western Australia, Perth, 6009, Australia}
\begin{abstract}
An important class of dynamical systems with several practical applications is linear systems with quadratic outputs. These models have the same state equation as standard linear time-invariant systems but differ in their output equations, which are nonlinear quadratic functions of the system states. When dealing with models of exceptionally high order, the computational demands for simulation and analysis can become overwhelming. In such cases, model order reduction proves to be a useful technique, as it allows for constructing a reduced-order model that accurately represents the essential characteristics of the original high-order system while significantly simplifying its complexity.

In time-limited model order reduction, the main goal is to maintain the output response of the original system within a specific time range in the reduced-order model. To assess the error within this time interval, a mathematical expression for the time-limited $\mathcal{H}_2$-norm is derived in this paper. This norm acts as a measure of the accuracy of the reduced-order model within the specified time range. Subsequently, the necessary conditions for achieving a local optimum of the time-limited $\mathcal{H}_2$ norm error are derived. The inherent inability to satisfy these optimality conditions within the Petrov-Galerkin projection framework is also discussed. After that, a stationary point iteration algorithm based on the optimality conditions and Petrov-Galerkin projection is proposed. Upon convergence, this algorithm fulfills three of the four optimality conditions. To demonstrate the effectiveness of the proposed algorithm, a numerical example is provided that showcases its ability to effectively approximate the original high-order model within the desired time interval.
\end{abstract}


\end{frontmatter}

\section{Introduction}
This research work studies a special class of nonlinear dynamical systems characterized by weak nonlinearity. These systems retain linear time-invariant (LTI) state equations but incorporate quadratic nonlinear terms in their output equations, thus named linear quadratic output (LQO) systems \citep{montenbruck2017linear}. They naturally arise in scenarios requiring the observation of quantities that involve the product of state components, either in the time or frequency domain. These systems are valuable in quantifying energy or power of dynamical systems, for example, in assessing a system's internal energy \citep{van2006port} or the cost function in optimal quadratic control problems \citep{picinbono1988optimal}. Moreover, they are used to measure deviations of state coordinates from a reference point, such as for calculating the root mean squared displacement of spatial coordinates around an excitation point or for estimating the variance of a random variable in stochastic modeling \citep{aumann2023structured}.

Consider an LQO system defined by the following state and output equations:
\begin{align}
 H:=
  \begin{cases}
   \dot{x}(t)=Ax(t)+Bu(t),\\
  y(t)=Cx(t)+\begin{bmatrix}x(t)^TM_1x(t)\\\vdots\\x(t)^TM_px(t)\end{bmatrix},
  \end{cases}\label{steq:1}
\end{align} where $x(t)\in\mathbb{R}^{N\times 1}$ denotes the state vector, $u(t)\in\mathbb{R}^{m\times 1}$ represents inputs, $y(t)\in\mathbb{R}^{p\times 1}$ is the output vector. The matrices $A\in\mathbb{R}^{N\times N}$, $B\in\mathbb{R}^{N\times m}$, $C\in\mathbb{R}^{p\times N}$, and $M_i\in\mathbb{R}^{N\times N}$ consititute the state-space realization $(A,B,C,M_1,\cdots,M_p)$. The state equation is identical in form to that of a standard LTI system. However, the output equation introduces a nonlinearity through the quadratic terms $x(t)^TM_ix(t)$, which distinguishes the LQO system from a standard LTI system.

Accurately modeling complex physical phenomena frequently necessitates dynamical systems of exceptionally high order, often in the range of several thousand or more. Due to this very high order $N$, simulating and analyzing the model (\ref{steq:1}) becomes computationally intensive and impractical. Consequently, it becomes necessary to approximate (\ref{steq:1}) with a reduced-order model (ROM) of significantly lower order $n$ (where $n\ll N)$, which simplifies simulation and analysis \citep{van2012model}. Model order reduction (MOR) refers to the process of constructing a ROM while ensuring that the essential features and characteristics of the original model are preserved \citep{benner2005dimension}.

We define the $n^{th}$-order ROM of the system $H$ as $\hat{H}$, which is described by the state and output equations shown in \eqref{steq:2}
\begin{align}
 \hat{H}:=
  \begin{cases}
   \dot{x_n}(t)=\hat{A}x_n(t)+\hat{B}u(t),\\
  y_n(t)=\hat{C}x_n(t)+\begin{bmatrix}x_n(t)^T\hat{M_1}x_n(t)\\\vdots\\x_n(t)^T\hat{M_p}x_n(t)\end{bmatrix}.
  \end{cases}\label{steq:2}
\end{align} Here, $\hat{A}$, $\hat{B}$, $\hat{C}$, and $\hat{M}_i$ are matrices derived from the original system matrices $A$, $B$, $C$, and $M_i$, respectively, through the Petrov-Galerkin projection condition
$\hat{W}^T\hat{V}=I$. Specifically, $\hat{A}=\hat{W}^TA\hat{V}$, $\hat{B}=\hat{W}^TB$, $\hat{C}=C\hat{V}$, and $\hat{M_i}=\hat{V}^TM_i\hat{V}$. The projection matrices $\hat{V}\in\mathbb{R}^{N\times n}$ and $\hat{W}\in\mathbb{R}^{N\times n}$ are used to project the original system $H$ onto a reduced subspace, resulting in the ROM $\hat{H}$. Different MOR techniques vary in their approach to constructing $\hat{V}$ and $\hat{W}$, which is dependent on the specific characteristics of $H$ that need to be preserved in $\hat{H}$ \citep{antoulas2004model,antoulas2015model,antoulas2020interpolatory}. Throughout this paper, it is assumed that both $A$ and $\hat{A}$ are Hurwitz matrices.

The Balanced Truncation (BT) method, developed in 1981, is a widely utilized MOR technique \citep{moore1981principal}. This method selectively retains states that significantly contribute to the energy transfer between inputs and outputs, while discarding those with minimal influence, as determined by their Hankel singular values. One key advantage of BT is its capability to estimate errors \textit{a priori} before constructing the ROM \citep{enns1984model}. Furthermore, BT ensures that the stability of the original model is maintained. Initially developed for standard LTI systems, BT has significantly expanded its applicability to encompass various system types, including descriptor systems \citep{mehrmann2005balanced}, second-order systems \citep{reis2008balanced}, linear time-varying systems \citep{sandberg2004balanced}, parametric systems \citep{son2021balanced}, nonlinear systems \citep{kramer2022balanced}, and bilinear systems \citep{benner2017truncated}, among others. Additionally, BT has been adapted to preserve specific system properties, such as positive realness \citep{wong2007fast}, bounded realness \citep{guiver2013bounded}, passivity \citep{wong2004passivity}, and special structural characteristics \citep{sarkar2023structure}. For a comprehensive understanding of the diverse BT algorithm family, refer to the survey \citep{gugercin2004survey}. BT has been extended to LQO systems in \citep{van2010model,pulch2019balanced,benner2021gramians}. Among these algorithms, only the one presented in \citep{benner2021gramians} maintains the LQO structure in the ROM.

A locally optimal solution in the $\mathcal{H}_2$ norm is achievable using several efficient algorithms. The necessary conditions for this local optimum, known as Wilson's conditions \citep{wilson1970optimum}, involve interpolation of the original system at selected points. The iterative rational Krylov algorithm (IRKA) is a well-known algorithm that can achieve this local optimum \citep{gugercin2008h_2}. A more general algorithm, which does not assume the original system having simple poles, is also available \citep{xu2011optimal}. This algorithm is based on Sylvester equations and its numerical properties have been improved in \citep{MPIMD11-11}. Recently, the $\mathcal{H}_2$ MOR problem for LQO systems has been addressed, and an algorithm based on Sylvester equations has been proposed as a solution \citep{reiter2024h2}.

BT aims to approximate the original model dynamics over the entire time horizon. However, practical systems and simulations often operate within finite time intervals due to real-world constraints. For example, in interconnected power systems, low-frequency oscillations typically last around 15 seconds and are mitigated by power system stabilizers and damping controllers \citep{kundur2007power}. The first $15$ seconds are crucial for small-signal stability assessments \citep{rogers2000nature}. Similarly, in finite-time optimal control problems, the system's behavior within a specific time frame is of utmost importance \citep{grimble1979solution}. This need has led to the development of time-limited MOR, which focuses on achieving good accuracy within a specified time interval rather than the entire time domain \citep{gawronski1990model}.

The time-limited MOR problem focuses on developing a ROM that guarantees a small deviation between the outputs of the original model $ y(t) $ and the ROM's output $ y_n(t) $ for a given input signal $ u(t)$, but only within a specified limited time interval $[0, \tau]$ seconds. In other words, the goal is to ensure that the approximation error $ \| y(t) - y_n(t) \| $ remains minimal within this specified time frame \citep{zulfiqar2024relative}.

To address the time-limited MOR problem, BT has been adapted into the time-limited BT (TLBT) algorithm \citep{gawronski1990model}. Although TLBT does not fully retain all characteristics of traditional BT, such as stability guarantees and \textit{a priori} error bounds, it effectively addresses the time-limited MOR scenario. The computational aspects of TLBT and strategies for handling large-scale systems have been explored in \citep{kurschner2018balanced}. Furthermore, TLBT has been extended to handle descriptor systems \citep{kurschner2018balanced}, second-order systems \citep{benner2021frequency}, and bilinear systems \cite{shaker2014time}, enhancing its applicability. Recently, TLBT has also been extended to LQO systems in \citep{song2024balanced}.

In \citep{goyal2019time}, a new norm called the time-limited $\mathcal{H}_2$ norm is defined, and the conditions needed to find a local optimum for this norm are derived. An algorithm called time-limited IRKA (TLIRKA) is then proposed based on IRKA to construct a local optimum, but it does not meet any necessary conditions. In \citep{sinani2019h2}, new optimality conditions based on interpolation are derived, and a nonlinear optimization algorithm is proposed to construct the local optimum, whose applicability is limited to medium-scale systems. This paper studies the time-limited $\mathcal{H}_2$-optimal MOR problem for LQO systems.

The key contributions of this research work are manifold. First, it defines the time-limited $\mathcal{H}_2$ norm ($\mathcal{H}_{2,\tau}$ norm) for LQO systems and demonstrates its computation using time-limited system Gramians defined in \citep{song2024balanced}. Second, it derives the necessary condition for achieving a local optimum of $||H-\hat{H}||_{\mathcal{H}_{2,\tau}}^2$. Third, these conditions are then contrasted with those related to the standard $\mathcal{H}_2$-optimal MOR \citep{reiter2024h2}. Notably, it is shown that Petrov-Galerkin projection cannot, in general, attain a local optimum in the time-limited setting. Fourth, a stationary point algorithm, rooted in the Petrov-Galerkin projection, is proposed. Upon convergence, this algorithm satisfies three of the four necessary conditions for optimality. An illustrative numerical example is presented to demonstrate the accuracy of the proposed algorithm within the specified time interval.
\section{Literature Review}
In this section, we will briefly discuss two of the most relevant MOR algorithms for LQO systems in the context of the problem under consideration. The first algorithm is the $\mathcal{H}_2$-optimal MOR method \citep{reiter2024h2}, and the second is the TLBT \citep{song2024balanced}.
\subsection{$\mathcal{H}_2$-optimal MOR Algorithm (HOMORA) \citep{reiter2024h2}}
Let us define the matrices $\tilde{P}$, $\hat{P}$, $\tilde{Y}$, $\hat{Y}$, $\tilde{Z}$, $\hat{Z}$, $\tilde{Q}$, and $\hat{Q}$ that solve the following set of linear matrix equations:
\begin{align}
A\tilde{P}+\tilde{P}\hat{A}^T+B\hat{B}^T=0,\nonumber\\
\hat{A}\hat{P}+\hat{P}\hat{A}^T+\hat{B}\hat{B}^T=0,\nonumber\\
A^T\tilde{Y}+\tilde{Y}\hat{A}+C^T\hat{C}=0,\nonumber\\
\hat{A}^T\hat{Y}+\hat{Y}\hat{A}+\hat{C}^T\hat{C}=0,\nonumber\\
A^T\tilde{Z}+\tilde{Z}\hat{A}+\sum_{i=1}^{p}M_i\tilde{P}\hat{M_i}=0,\nonumber\\
\hat{A}^T\hat{Z}+\hat{Z}\hat{A}+\sum_{i=1}^{p}\hat{M_i}\hat{P}\hat{M_i}=0,\nonumber\\
A^T\tilde{Q}+\tilde{Q}\hat{A}+C^T\hat{C}+\sum_{i=1}^{p}M_i\tilde{P}\hat{M_i}=0,\nonumber\\
\hat{A}^T\hat{Q}+\hat{Q}\hat{A}+\hat{C}^T\hat{C}+\sum_{i=1}^{p}\hat{M_i}\hat{P}\hat{M_i}=0.\nonumber
\end{align}
According to \citep{reiter2024h2}, the necessary conditions for the local optimum of the (squared) $\mathcal{H}_2$-norm of the error denoted as $||H - \hat{H}||_{\mathcal{H}_2}^2$, are given by the following set of equations:
\begin{align}
-(\tilde{Y}+2\tilde{Z})^T\tilde{P}+(\hat{Y}+2\hat{Z})\hat{P}&=0,\label{op01}\\
-\tilde{P}^TM_i\tilde{P}+\hat{P}\hat{M_i}\hat{P}&=0,\label{op02}\\
-(\tilde{Y}+2\tilde{Z})^TB+(\hat{Y}+2\hat{Z})\hat{B}&=0,\label{op03}\\
-C\tilde{P}+\hat{C}\hat{P}&=0.\label{op04}
\end{align}
Furthermore, it is shown that these optimality conditions can be achieved by setting the projection matrices as $\hat{V} = \tilde{P}$ and $\hat{W} = (\tilde{Y} + 2\tilde{Z})\big(\tilde{P}^T(\tilde{Y} + 2\tilde{Z})\big)^{-1}$. Starting with an initial guess of the ROM, the projection matrices are iteratively updated until convergence is reached, at which point the optimality conditions (\ref{op01})-(\ref{op04}) are satisfied.
\subsection{Time-limited Balanced Truncation (TLBT) \citep{song2024balanced}}
TLBT constructs the ROM by identifying and truncating the states that have minimal contribution to the input-output energy transfer within the desired time interval $[0, \tau]$ seconds. This is achieved by first constructing a time-limited balanced realization using time-limited Gramians, and then truncating the states that correspond to the smallest time-limited Hankel singular values.

The time-limited controllability Gramian $P_\tau$ within the desired time interval $[0,\tau]$ seconds is defined as
\begin{align}
P_\tau=\int_{0}^{\tau}e^{At}BB^Te^{A^Tt}dt.\nonumber
\end{align}
By defining $S_\tau$ as $S_\tau = e^{A\tau}$, $P_\tau$ can be found by solving the following Lyapunov equation:
\begin{align}
AP_\tau+P_\tau A^T+BB^T-S_{\tau}BB^TS_{\tau}^T&=0.\nonumber
\end{align}
The time-limited observability Gramian $Q_\tau$ within the time interval $[0,\tau]$ sec is defined as the sum of $Y_\tau$ and $Z_\tau$ where $Y_\tau$ and $Z_\tau$ are given by
\begin{align}
Y_{\tau}&=\int_{0}^{\tau}e^{A^Tt}C^TCe^{At}dt,\nonumber\\
Z_{\tau}&=\int_{0}^{\tau}e^{A^Tt_1}\Bigg(\sum_{i=1}^{p}M_i\Big(\int_{0}^{\tau}e^{At_2}BB^Te^{A^Tt_2}dt_2\Big)M_i\Bigg)e^{At_1}dt_1.\nonumber
\end{align} $Y_{\tau}$, $Z_{\tau}$, and $Q_\tau$ can be computed through the solution of following Lyapunov equations:
\begin{align}
A^TY_{\tau}+Y_{\tau} A+C^TC-S_{\tau}^TC^TCS_{\tau}&=0,\nonumber\\
A^TZ_{\tau}+Z_{\tau} A+\sum_{i=1}^{p}\big(M_iP_\tau M_i-S_{\tau}^TM_iP_\tau M_iS_{\tau}\big)&=0,\nonumber\\
A^TQ_\tau+Q_\tau A+C^TC-S_{\tau}^TC^TCS_{\tau}+\sum_{i=1}^{p}\big(M_iP_\tau M_i-S_{\tau}^TM_iP_\tau M_iS_{\tau}\big)&=0.\nonumber
\end{align}
The time-limited Hankel singular values, denoted by $\sigma_i$, are computed as the square root of the eigenvalues of the product of matrices $P_\tau$ and $Q_\tau$:
\begin{align}
\sigma_i=\sqrt{\lambda_i(P_\tau Q_\tau)}\hspace*{0.5cm}\textnormal{for} \hspace*{0.5cm}i=1,\cdots,N,\nonumber
\end{align}wherein $\lambda_i(\cdot)$ represents the eigenvalues. The projection matrices in TLBT are determined such that the transformed matrices $\hat{W}^TP_\tau \hat{W}$ and $\hat{V}^TQ_\tau \hat{V}$ become diagonal matrices with the $n$ largest time-limited Hankel singular values of $H$, i.e., $\hat{W}^TP_\tau \hat{W}=\hat{V}^TQ_\tau \hat{V}=diag(\sigma_1,\cdots,\sigma_n)$.
\section{Main Work}
This section introduces the time-limited $\mathcal{H}_2$ norm and its connection to the time-limited observability Gramian. Subsequently, necessary conditions for the local optimality of the squared time-limited $\mathcal{H}_2$ norm of the error are derived. Based on these optimality conditions, a Petrov-Galerkin projection-based iterative algorithm is proposed that satisfies three of the four optimality conditions. The challenges associated with fulfilling the remaining optimality condition within the Petrov-Galerkin projection framework are also discussed. Finally, the computational efficiency of the proposed algorithm is briefly analyzed.
\subsection{$\mathcal{H}_{2,\tau}$ norm Definition}
The classical $\mathcal{H}_2$ norm for LQO systems is defined in the time domain as:
\begin{align}
||H||_{\mathcal{H}_2}&=\Bigg[trace\Big(\int_{0}^{\infty}h_1^T(t)h_1(t)dt\nonumber\\
&\hspace*{1.5cm}+\int_{0}^{\infty}\int_{0}^{\infty}\sum_{i=1}^{p}h_{2,i}^T(t_1,t_2)h_{2,i}(t_1,t_2)dt_1dt_2\Big)\Bigg]^{-\frac{1}{2}},\nonumber
\end{align}where $h_1(t)=Ce^{At}B$ and $h_{2,i}(t_1,t_2)=B^Te^{A^Tt_1}M_ie^{At_2}B$. The $\mathcal{H}_2$ norm measures the output power in response to unit white noise over the entire time horizon. However, for the problem at hand, we focus on the output power within a specific time interval, leading to the concept of the time-limited $\mathcal{H}_2$ norm.
\begin{definition}
The time-limited $\mathcal{H}_2$ norm of the LQO system within the time interval $[0,\tau]$ sec is defined as:
\begin{align}
||H||_{\mathcal{H}_{2,\tau}}&=\Big[trace\Big(\int_{0}^{\tau}h_1^T(t)h_1(t)dt\nonumber\\
&\hspace*{1.5cm}+\int_{0}^{\tau}\int_{0}^{\tau}\sum_{i=1}^{p}h_{2,i}^T(t_1,t_2)h_{2,i}(t_1,t_2)dt_1dt_2\Big)\Big]^{-\frac{1}{2}}.\nonumber
\end{align}
\end{definition}
\begin{proposition}
The $\mathcal{H}_{2,\tau}$ norm is related to the time-limited observability Gramian $Q_\tau$ as follows:
\begin{align}
||H||_{\mathcal{H}_{2,\tau}}=\sqrt{trace(B^TQ_\tau B)}.\nonumber
\end{align}
\end{proposition}
\begin{proof}
Observe that
\begin{align}
&trace\Big(\int_{0}^{\tau}h_1^T(t)h_1(t)dt\Big)\nonumber\\
&=trace\Big(B^T\Big[\int_{0}^{\tau}e^{A^Tt}C^TCe^{At}\Big]Bdt\Big)\nonumber\\
&=trace(B^TY_{\tau} B).\nonumber
\end{align}
Furthermore, notice that
\begin{align}
&trace\Big(\int_{0}^{\tau}\int_{0}^{\tau}\sum_{i=1}^{p}h_{2,i}^T(t_1,t_2)h_{2,i}(t_1,t_2)dt_1dt_2\Big)\nonumber\\
&=trace\Bigg(B^T\Big[\int_{0}^{\tau}e^{A^Tt_2}\Big(\sum_{i=1}^{p}M_i\Big(\int_{0}^{\tau}e^{At_1}BB^Te^{A^Tt_1}dt_1\Big)M_i\Big)e^{At_2}dt_2\Big]B\Bigg)\nonumber\\
&=trace(B^TZ_{\tau} B).\nonumber
\end{align}
Therefore, $||H||_{\mathcal{H}_{2,\tau}}=\sqrt{trace\big(B^T(Y_{\tau}+Z_{\tau})B\big)}=\sqrt{trace\big(B^T(Q_\tau)B\big)}$.
\end{proof}
\subsection{$\mathcal{H}_{2,\tau}$ Norm of the Error}
Let us define the error system $E=H-\hat{H}$ with the following state-space representation:
\begin{align}
 E:=
  \begin{cases}
   \dot{x_e}(t)=\begin{bmatrix}x(t)\\x_n(t)\end{bmatrix}=A_ex_e(t)+B_eu(t),\\
  y_e(t)=y(t)-y_n(t)=C_ex_e(t)+\begin{bmatrix}x_e(t)^TM_{e,1}x_e(t)\\\vdots\\x_e(t)^TM_{e,p}x_e(t)\end{bmatrix},
  \end{cases}\nonumber
\end{align}where
\begin{align}
A_e&=\begin{bmatrix}A&0\\0&\hat{A}\end{bmatrix},&B_e&=\begin{bmatrix}B\\\hat{B}\end{bmatrix},&M_{e,i}&=\begin{bmatrix}M_i&0\\0&-\hat{M_i}\end{bmatrix},&C_e&=\begin{bmatrix}C&-\hat{C}\end{bmatrix}.\label{partreal}
\end{align}
Let $S_{e,\tau}$ be defined as $S_{e,\tau}=e^{A_e\tau}$. The time-limited controllability Gramian $P_{e,\tau}$ and the time-limited observability Gramian $Q_{e,\tau}=Y_{e,\tau}+Z_{e,\tau}$ of the realization $(A_e,B_e,C_e,M_{e,1},\cdots,M_{e,p})$ can be obtained by solving the following Lyapunov equations:
\begin{align}
\hspace*{1cm}A_eP_{e,\tau}+P_{e,\tau} A_e^T&+  B_eB_e^T-S_{e,\tau}B_eB_e^TS_{e,\tau}^T=0,\nonumber\\
A_e^TY_{e,\tau}+Y_{e,\tau} A_e&+C_e^TC_e-S_{e,\tau}^TC_e^TC_eS_{e,\tau}=0,\nonumber\\
A_e^TZ_{e,\tau}+Z_{e,\tau} A_e&+\sum_{i=1}^{p}\big(M_{e,i}P_{e,\tau} M_{e,i}\nonumber\\
&\hspace*{0.5cm}-S_{e,\tau}^TM_{e,i}P_{e,\tau} M_{e,i}S_{e,\tau}\big)=0,\nonumber\\
A_e^TQ_{e,\tau}+Q_{e,\tau} A_e&+C_e^TC_e-S_{e,\tau}^TC_e^TC_eS_{e,\tau}\nonumber\\
&\hspace*{0.5cm}+\sum_{i=1}^{p}\big(M_{e,i}P_{e,\tau} M_{e,i}-S_{e,\tau}^TM_{e,i}P_{e,\tau} M_{e,i}S_{e,\tau}\big)=0.\nonumber
\end{align}
Partition $P_{e,\tau}$, $Y_{e,\tau}$, $Z_{e,\tau}$, and $Q_{e,\tau}$ according to (\ref{partreal}) as follows:
\begin{align}
P_{e,\tau}&=\begin{bmatrix}P_\tau&\tilde{P}_{\tau}\\\tilde{P}_{\tau}^T&\hat{P}_{\tau}\end{bmatrix},&
Y_{e,\tau}&=\begin{bmatrix}Y_{\tau}&-\tilde{Y}_{\tau}\\-\tilde{Y}_{\tau}^T&\hat{Y}_{\tau}\end{bmatrix},\nonumber\\
Z_{e,\tau}&=\begin{bmatrix}Z_{\tau}&-\tilde{Z}_{\tau}\\-\tilde{Z}_{\tau}^T&\hat{Z}_{\tau}\end{bmatrix},&
Q_{e,\tau}&=\begin{bmatrix}Q_\tau&-\tilde{Q}_{\tau}\\-\tilde{Q}_{\tau}^T&\hat{Q}_{\tau}\end{bmatrix}.\nonumber
\end{align}
By defining $\hat{S}_{\tau}=e^{\hat{A}\tau}$, it can be verified that the following linear matrix equations hold:
\begin{align}
A\tilde{P}_{\tau}+\tilde{P}_{\tau}\hat{A}^T&+ B\hat{B}^T-S_{\tau}B\hat{B}^T\hat{S}_{\tau}^T=0,\label{eq:24}\\
\hat{A}\hat{P}_{\tau}+\hat{P}_{\tau}\hat{A}^T&+\hat{B}\hat{B}^T-\hat{S}_{\tau}\hat{B}\hat{B}^T\hat{S}_{\tau}^T=0,\label{eq:25}\\
A^T\tilde{Y}_{\tau}+\tilde{Y}_{\tau}\hat{A}&+C^T\hat{C}-S_{\tau}^TC^T\hat{C}\hat{S}_{\tau}=0,\label{eq:26}\\
\hat{A}^T\hat{Y}_{\tau}+\hat{Y}_{\tau}\hat{A}&+\hat{C}^T\hat{C}-\hat{S}_{\tau}^T\hat{C}^T\hat{C}\hat{S}_{\tau}=0,\label{eq:27}\\
A^T\tilde{Z}_{\tau}+\tilde{Z}_{\tau}\hat{A}&+\sum_{i=1}^{p}\big(M_i\tilde{P}_{\tau}\hat{M_i}\nonumber\\
&-S_{\tau}^T M_i\tilde{P}_{\tau}\hat{M_i}\hat{S}_{\tau}\big)=0,\label{eq:28}\\
\hat{A}^T\hat{Z}_{\tau}+\hat{Z}_{\tau}\hat{A}&+\sum_{i=1}^{p}\big(\hat{M_i}\hat{P}_{\tau}\hat{M_i}\nonumber\\
&-\hat{S}_{\tau}^T\hat{M_i}\hat{P}_{\tau}\hat{M_i}\hat{S}_{\tau}\big)=0,\label{eq:29}\\
A^T\tilde{Q}_{\tau}+\tilde{Q}_{\tau}\hat{A}&+C^T\hat{C}-S_{\tau}^TC^T\hat{C}\hat{S}_{\tau}\nonumber\\
&+\sum_{i=1}^{p}\big(M_i\tilde{P}_{\tau}\hat{M_i}-S_{\tau}^TM_i\tilde{P}_{\tau}\hat{M_i}\hat{S}_{\tau}\big)=0,\label{eq:30}\\
\hat{A}^T\hat{Q}_{\tau}+\hat{Q}_{\tau}\hat{A}&+\hat{C}^T\hat{C}-\hat{S}_{\tau}^T\hat{C}^T\hat{C}\hat{S}_{\tau}\nonumber\\
&+\sum_{i=1}^{p}\big(\hat{M_i}\hat{P}_{\tau}\hat{M_i}-\hat{S}_{\tau}^T\hat{M_i}\hat{P}_{\tau}\hat{M_i}\hat{S}_{\tau}\big)=0.\label{eq:31}
\end{align}
Consequently, the $\mathcal{H}_{2,\tau}$ norm of $E$ can be expressed as:
\begin{align}
||E||_{\mathcal{H}_{2,\tau}}&=\sqrt{trace(B_e^TQ_{e,\tau}B_e)}\nonumber\\
&=\sqrt{trace(B^TQ_\tau B-2B^T\tilde{Q}_{\tau}\hat{B}+\hat{B}^T\hat{Q}_{\tau}\hat{B})}.\nonumber
\end{align}
\begin{corollary}
$||E||_{\mathcal{H}_{2,\tau}}^2=||H||_{\mathcal{H}_{2,\tau}}^2-2\langle  H , \hat{H} \rangle_{\mathcal{H}_{2,\tau}}+||\hat{H}||_{\mathcal{H}_{2,\tau}}^2,$ wherein $\langle  H , \hat{H} \rangle_{\mathcal{H}_{2,\tau}}$ represents the $\mathcal{H}_{2,\tau}$ inner product of $H$ and $\hat{H}$.
\end{corollary}
\begin{proof}
The first and last terms in the expression for $||E||_{\mathcal{H}_{2,\tau}}^2$ are straightforward. Our focus lies in showing that the middle term corresponds to the $\mathcal{H}_{2,\tau}$ inner product of $H$ and $\hat{H}$. By expanding the inner product definition, we can express it as:
\begin{align}
\langle  H , \hat{H} \rangle_{\mathcal{H}_{2,\tau}}&=\int_{0}^{\tau}trace\big(h_1^T(t)\hat{h}_1(t)dt\big)\nonumber\\
&\hspace*{1cm}+trace\Big(\int_{0}^{\tau}\int_{0}^{\tau}\sum_{i=1}^{p}\big(h_{2,i}^T(t_1,t_2)\hat{h}_{2,i}(t_1,t_2)\big)dt_1dt_2\Big),\nonumber
\end{align}wherein $\hat{h}_1(t)=\hat{C}e^{\hat{A}t}\hat{B}$ and $\hat{h}_{2,i}(t_1,t_2)=\hat{B}^Te^{\hat{A}^Tt_1}\hat{M}_ie^{\hat{A}t_2}\hat{B}$.
Furthermore,
\begin{align}
trace\Big(\int_{0}^{\tau}h_1^T(t)\hat{h}_1(t)dt\Big)&=\nonumber\\
trace\Big(B^T\big(\int_{0}^{\tau}&e^{A^Tt}C^T\hat{C}e^{\hat{A}}dt\big)\hat{B}\Big),\nonumber\\
trace\Big(\int_{0}^{\tau}\int_{0}^{\tau}\sum_{i=1}^{p}h_{2,i}^T(t_1,t_2)\hat{h}_{2,i}(t_1,t_2)dt_1dt_2\Big)&=\nonumber\\
trace\Bigg(B^T\Big[\int_{0}^{\tau}e^{A^Tt_1}\Big(\sum_{i=1}^{p}M_i\big(\int_{0}^{\tau}&e^{At_2}B\hat{B}^Te^{\hat{A}^Tt_2}dt_2\big)\hat{M}_i\Big)e^{\hat{A}t_1}\Big]\hat{B}\Bigg),\nonumber
\end{align}
The Sylvester equations (\ref{eq:26}) and (\ref{eq:28}) can be solved by computing the following integrals:
\begin{align}
\tilde{Y}_{\tau}&=\int_{0}^{\tau}e^{A^Tt}C^T\hat{C}e^{\hat{A}t}dt,\nonumber\\
\tilde{Z}_{\tau}&=\int_{0}^{\tau}e^{A^Tt_1}\Bigg(\sum_{i=1}^{p}M_i\Big(\int_{0}^{\tau}e^{At_2}B\hat{B}^Te^{\hat{A}^Tt_2}dt_2\Big)\hat{M}_i\Bigg)e^{\hat{A}t_1}dt_1;\nonumber
\end{align}cf. \citep{benner2021gramians,song2024balanced} Consequently, the $\mathcal{H}_{2,\tau}$ inner product between $H$ and $\hat{H}$ can be expressed as
\begin{align}
\langle  H , \hat{H} \rangle_{\mathcal{H}_{2,\tau}}=trace\big(B^T(\tilde{Y}_\tau+\tilde{Z}_\tau)\hat{B}\big)=trace(B^T\tilde{Q}_\tau\hat{B}).\nonumber
\end{align}
\end{proof}
\subsection{Optimality Conditions}
This subsection establishes the necessary conditions for minimizing the squared $\mathcal{H}_{2,\tau}$-norm of the error. To achieve this, several auxiliary variables are introduced. Specifically, $\bar{Z}_{\tau}$ and $\bar{Z}_{n,\tau}$ are defined as the solutions to the following linear matrix equations:
\begin{align}
A^T\bar{Z}_{\tau}+\bar{Z}_{\tau}\hat{A}+\sum_{i=1}^{p}M_i\tilde{P}_{\tau}\hat{M_i}=0,\nonumber\\
\hat{A}^T\bar{Z}_{n,\tau}+\bar{Z}_{n,\tau}\hat{A}+\sum_{i=1}^{p}\hat{M_i}\hat{P}_{\tau}\hat{M_i}=0.\nonumber
\end{align}
Note that $\tilde{P}_{\tau}$, $\hat{P}_{\tau}$, $\tilde{Z}_{\tau}$, and $\hat{Z}_{\tau}$ are truncated versions of $\tilde{P}$, $\hat{P}$, $\bar{Z}_{\tau}$, and $\bar{Z}_{n,\tau}$, respectively, with integration limits restricted to $[0,\tau]$. Additionally, we define $\tilde{P}_{12}$, $\tilde{P}_n$, $\tilde{Z}_{12}$, and $\tilde{Z}_n$ as follows:
\begin{align}
\tilde{P}_{12}&=\tilde{P}\Big|_{\tau}^{\infty}=\tilde{P}\Big|_{0}^{\infty}-\tilde{P}\Big|_{0}^{\tau}=\tilde{P}-\tilde{P}_{\tau},\label{nst46}\\
\tilde{P}_n&=\hat{P}\Big|_{\tau}^{\infty}=\hat{P}\Big|_{0}^{\infty}-\hat{P}\Big|_{0}^{\tau}=\hat{P}-\hat{P}_{\tau},\label{nst47}\\
\tilde{Z}_{12}&=\bar{Z}_{\tau}\Big|_{\tau}^{\infty}=\bar{Z}_{\tau}\Big|_{0}^{\infty}-\bar{Z}_{\tau}\Big|_{0}^{\tau}=\bar{Z}_{\tau}-\tilde{Z}_{\tau},\label{nst48}\\
\tilde{Z}_n&=\bar{Z}_{n,\tau}\Big|_{\tau}^{\infty}=\bar{Z}_{n,\tau}\Big|_{0}^{\infty}-\bar{Z}_{n,\tau}\Big|_{0}^{\tau}=\bar{Z}_{n,\tau}-\hat{Z}_{\tau}.\label{nst49}
\end{align}
Furthermore, let us define $V$, $W$, and $L_\tau$ as
\begin{align}
V&=\hat{B}B^TS_{\tau}^T\bar{Z}_{\tau}-\hat{B}\hat{B}^T\hat{S}_{\tau}^T\bar{Z}_{n,\tau}+\tilde{P}^TS_{\tau}^TC^T\hat{C}-\hat{P}\hat{S}_{\tau}^T\hat{C}\hat{C}^T\nonumber\\
&\hspace*{1.5cm}+\sum_{i=1}^{p}\big(\tilde{P}^TS_{\tau}^TM_i\tilde{P}_{\tau}\hat{M_i}-\hat{P}\hat{S}_{\tau}^T\hat{M_i}\hat{P}_{\tau}\hat{M_i}\big)\nonumber\\
W&=\mathcal{L}(\hat{A},V)=\int_{0}^{\tau}e^{\hat{A}(\tau-t)}Ve^{(\hat{A}+V)t}dt,\nonumber\\
L_\tau&=-\tilde{Q}_{\tau}^T\tilde{P}_{12}+\hat{Q}_\tau\tilde{P}_n-\tilde{Z}_{12}^T\tilde{P}_\tau+\tilde{Z}_n\tilde{P}_\tau+W^T,\nonumber
\end{align}where $\mathcal{L}(\hat{A},V)$ represents the Fr\'{e}chet derivative of the matrix exponential $e^{\hat{A}\tau}$ in the direction of matrix $V$.

We now present the necessary conditions for achieving a local minimum of the squared $\mathcal{H}_{2,\tau}$-norm of the error.
\begin{theorem}\label{th1}
A local minimum of $||E||_{\mathcal{H}_{2,\tau}}^2$ must satisfy the following necessary conditions:
\begin{align}
-(\tilde{Y}_{\tau}+2\tilde{Z}_{\tau})^T\tilde{P}_{\tau}+(\hat{Y}_{\tau}+2\hat{Z}_{\tau})\hat{P}_{\tau}+L_\tau&=0,\label{op1}\\
-\tilde{P}_{\tau}^TM_i\tilde{P}_{\tau}+\hat{P}_{\tau}\hat{M_i}\hat{P}_{\tau}&=0,\label{op2}\\
-(\tilde{Y}_{\tau}+2\tilde{Z}_{\tau})^TB+(\hat{Y}_{\tau}+2\hat{Z}_{\tau})\hat{B}&=0,\label{op3}\\
-C\tilde{P}_{\tau}+\hat{C}\hat{P}_{\tau}&=0.\label{op4}
\end{align}
\end{theorem}
\begin{proof}
The proof of this theorem is tedious and lengthy and is therefore deferred to Appendix A.
\end{proof}
\subsection{Comparison with Local Optimum of $||E||_{\mathcal{H}_2}^2$}
This subsection draws a comparison between the necessary conditions for optimizing the standard $\mathcal{H}_2$-norm and the time-limited $\mathcal{H}_{2,\tau}$-norm of the error. To facilitate this comparison, we begin by mathematically expressing the standard $\mathcal{H}_2$-norm as presented in \citep{benner2021gramians,reiter2024h2}. The controllability Gramian, denoted as $P_e$, and the observability Gramian, denoted as $Q_e = Y_e + Z_e$, associated with the system realization $(A_e, B_e, C_e, M_{e,1}, \dots, M_{e,p})$ can be determined by solving the following Lyapunov equations:
\begin{align}
A_eP_e+P_e A_e^T+ B_eB_e^T&=0,\nonumber\\
A_e^TY_e+Y_e A_e+C_e^TC_e&=0,\nonumber\\
A_e^TZ_e+Z_e A_e+\sum_{i=1}^{p}M_{e,i}P_e M_{e,i}&=0,\nonumber\\
A_e^TQ_e+Q_e A_e+C_e^TC_e+\sum_{i=1}^{p}M_{e,i}P_eM_{e,i}&=0.\nonumber
\end{align}
Partitioning $P_e$, $Y_e$, $Z_e$, and $Q_e$ according to the structure outlined in (\ref{partreal}) yields:
\begin{align}
P_e&=\begin{bmatrix}P&\tilde{P}\\\tilde{P}^T&\hat{P}\end{bmatrix},&
Y_e&=\begin{bmatrix}Y&-\tilde{Y}\\-\tilde{Y}^T&\hat{Y}\end{bmatrix},\nonumber\\
Z_e&=\begin{bmatrix}Z&-\tilde{Z}\\-\tilde{Z}^T&\hat{Z}\end{bmatrix},&
Q_e&=\begin{bmatrix}Q&-\tilde{Q}\\-\tilde{Q}^T&\hat{Q}\end{bmatrix}.\nonumber
\end{align}
Consequently, the $\mathcal{H}_2$-norm of $E$ can be expressed as:
\begin{align}
||E||_{\mathcal{H}_2}&=\sqrt{trace(B_e^TQ_eB_e)}\nonumber\\
&=\sqrt{trace(B^TQ B-2B^T\tilde{Q}\hat{B}+\hat{B}^T\hat{Q}\hat{B})}.\nonumber
\end{align}

A comparison of the optimality conditions (\ref{op1})-(\ref{op4}) and (\ref{op01})-(\ref{op04}) reveals both similarities and distinct differences. By restricting the integration limits of $P_e$ and $Q_e$ to the $[0, \tau]$ second, the optimality conditions (\ref{op2})-(\ref{op4}) can be derived from their counterparts (\ref{op02})-(\ref{op04}). However, the optimality condition (\ref{op01}) does not reduce to condition (\ref{op1}) through this integration limit constraint.

The optimality conditions (\ref{op02})-(\ref{op04}) directly yield optimal expressions for $\hat{M_i}$, $\hat{B}$, and $\hat{C}$ as follows:
\begin{align}
\hat{M_i}&=\hat{P}^{-1}\tilde{P}^TM_i\tilde{P}\hat{P}^{-1},\label{ocm0}\\
\hat{B}&=(\hat{Y}+2\hat{Z})^{-1}(\tilde{Y}+2\tilde{Z})^TB,\label{obm0}\\
\hat{C}&=C\tilde{P}\hat{P}^{-1}.\label{occ0}
\end{align}
By imposing the time constraint through integration limits on $P_e$ and $Q_e$, these optimal expressions can be adapted for the time-limited case, resulting in Equations (\ref{ocm1})-(\ref{occ1}), consistent with the optimality conditions (\ref{op2})-(\ref{op4}).
\begin{align}
\hat{M_i}&=\hat{P}_{\tau}^{-1}\tilde{P}_{\tau}^TM_i\tilde{P}_{\tau}\hat{P}_{\tau}^{-1},\label{ocm1}\\
\hat{B}&=(\hat{Y}_{\tau}+2\hat{Z}_{\tau})^{-1}(\tilde{Y}_{\tau}+2\tilde{Z}_{\tau})^TB,\label{obm1}\\
\hat{C}&=C\tilde{P}_{\tau}\hat{P}_{\tau}^{-1}.\label{occ1}
\end{align}
The optimal projection matrices $\hat{V}$ and $\hat{W}$ for the standard $\mathcal{H}_2$ optimal model order reduction problem can be defined as:
\begin{align}
\hat{V}&=\tilde{P}\hat{P}^{-1},&\hat{W}&=(\tilde{Y}+2\tilde{Z})(\hat{Y}+2\hat{Z})^{-1}.\nonumber
\end{align}
Analogous definitions can be used for the time-limited case by replacing the original matrices with their time-limited counterparts, i.e.,
\begin{align}
\hat{V}&=\tilde{P}_{\tau}\hat{P}_{\tau}^{-1},&\hat{W}&=(\tilde{Y}_{\tau}+2\tilde{Z}_{\tau})(\hat{Y}_{\tau}+2\hat{Z}_{\tau})^{-1}.\nonumber
\end{align} While this approach enables optimal selections for $\hat{M_i}$, $\hat{B}$, and $\hat{C}$ as per Equations (\ref{ocm1})-(\ref{occ1}), it falls short in determining an optimal $\hat{A}$. Enforcing the Petrov-Galerkin condition, $\hat{W}^T\hat{V} = I$, ensures that the term $-(\tilde{Y}_{\tau}+2\tilde{Z}_{\tau})^T\tilde{P}_{\tau}+(\hat{Y}_{\tau}+2\hat{Z}_{\tau})\hat{P}_{\tau}$ becomes zero, but the residual term $L_\tau$ generally remains nonzero, indicating a deviation from the optimality condition (\ref{op1}). Consequently, achieving a local optimum for the time-limited case within the Petrov-Galerkin projection framework remains generally impossible. Nevertheless, this choice of projection matrices yields exact satisfaction of the optimality conditions (\ref{op2})-(\ref{op4}) and an approximate satisfaction the optimality condition (\ref{op1}).

Thus far, the optimal projection matrices $\hat{V}$ and $\hat{W}$ have been determined as $\hat{V} = \tilde{P}_{\tau}\hat{P}_{\tau}^{-1}$ and $\hat{W} = (\tilde{Y}_{\tau}+2\tilde{Z}_{\tau})(\hat{Y}_{\tau}+2\hat{Z}_{\tau})^{-1}$, respectively. However, it is important to note that these matrices depend on the ROM, $(\hat{A}, \hat{B}, \hat{C}, \hat{M}_1, ..., \hat{M}_p)$, which is yet unknown. Consequently, Equations (\ref{steq:2}) and (\ref{eq:24})-(\ref{eq:29}) form a coupled system of equations that can be represented as:
\begin{align}
(\hat{A},\hat{B},\hat{C},\hat{M_1},\cdots,\hat{M_p})&=f(\tilde{P}_{\tau},\hat{P}_{\tau},\tilde{Y}_{\tau},\hat{Y}_{\tau},\tilde{Z}_{\tau},\hat{Z}_{\tau}),\nonumber\\
(\tilde{P}_{\tau},\hat{P}_{\tau},\tilde{Y}_{\tau},\hat{Y}_{\tau},\tilde{Z}_{\tau},\hat{Z}_{\tau})&=g(\hat{A},\hat{B},\hat{C},\hat{M_1},\cdots,\hat{M_p}).\nonumber
\end{align}
The stationary points of the composite function
\begin{align}
(\hat{A},\hat{B},\hat{C},\hat{M_1},\cdots,\hat{M_p})=f\big(g(\hat{A},\hat{B},\hat{C},\hat{M_1},\cdots,\hat{M_p})\big)\nonumber
\end{align} satisfy the optimality conditions (\ref{op2})-(\ref{op4}). Moreover, by imposing the Petrov-Galerkin condition, $\hat{W}^T\hat{V} = I$, the optimality condition (\ref{op1}) is approximately fulfilled with a deviation quantified by $L_\tau$. In contrast, the classical $\mathcal{H}_2$-optimal MOR problem achieves exact satisfaction of all optimality conditions (\ref{op01})-(\ref{op04}) through the enforcement of the Petrov-Galerkin condition on the stationary points.

In standard $\mathcal{H}_2$-optimal MOR, it is established that selecting projection matrices as $\hat{V} = \tilde{P}$ and $\hat{W} = \tilde{Y} + 2\tilde{Z}$, rather than the previously considered forms, leads to $\hat{P} = I$ and $\hat{Y} + 2\hat{Z} = I$ at stationary points. Consequently, the combination of these projection matrices with the Petrov-Galerkin condition, $\hat{W}^T\hat{V} = I$, ensures satisfaction of all optimality conditions (\ref{op01})-(\ref{op04}). However, when applying the analogous approach to the time-limited case, with $\hat{V} = \tilde{P}_{\tau}$ and $\hat{W} = \tilde{Y}_{\tau} + 2\tilde{Z}_{\tau}$, and imposing the Petrov-Galerkin condition, the resulting ROM does not satisfy any optimality condition.
\begin{theorem}\label{th2}
If the conditions $\hat{W}^TS_{\tau}B = \hat{S}_{\tau}\hat{B}$, $CS_{\tau}\hat{V} = \hat{C}\hat{S}_{\tau}$, and $\hat{V}^TM_iS_{\tau}\hat{V} = \hat{M_i}\hat{S}_{\tau}$ hold, then selecting $\hat{V} = \tilde{P}_{\tau}$ and $\hat{W} = \tilde{Y}_{\tau} + 2\tilde{Z}_{\tau}$ with the Petrov-Galerkin condition $\hat{W}^T\hat{V} = I$ leads to $\hat{P}_{\tau} = I$ and $\hat{Y}_{\tau} + 2\hat{Z}_{\tau} = I$, consequently satisfying the optimality conditions (\ref{op1})-(\ref{op4}).
\end{theorem}
\begin{proof}
The proof is detailed in Appendix B.
\end{proof}
In general, the chosen projection matrices, $\hat{V} = \tilde{P}_{\tau}$ and $\hat{W} = \tilde{Y}_{\tau} + 2\tilde{Z}_{\tau}$, along with the Petrov-Galerkin condition $\hat{W}^T\hat{V}=I$, do not satisfy the conditions $\hat{W}^TS_{\tau}B = \hat{S}_{\tau}\hat{B}$, $CS_{\tau}\hat{V} = \hat{C}\hat{S}_{\tau}$, and $\hat{V}^TM_iS_{\tau}\hat{V} = \hat{M_i}\hat{S}_{\tau}$. As a result, $\hat{P}_{\tau} \neq I$ and $\hat{Y}_{\tau} + 2\hat{Z}_{\tau} \neq I$ at the stationary points, and consequently, the optimality conditions (\ref{op2})-(\ref{op4}) are not fulfilled.
\subsection{Algorithm}
For simplicity, we have thus far assumed the desired time interval begins at 0 seconds. However, for any general time interval $[\tau_1, \tau_2]$ seconds, $\tilde{P}_\tau$, $\hat{P}_\tau$, $\tilde{Y}_\tau+2\tilde{Z}_\tau$, and $\hat{Y}_\tau+2\hat{Z}_\tau$ can be calculated by solving the following equations:
\begin{align}
A\tilde{P}_{\tau}+\tilde{P}_{\tau}\hat{A}^T+e^{A\tau_1}B\hat{B}^Te^{\hat{A}^T\tau_1}-e^{A\tau_2}B\hat{B}^Te^{\hat{A}^T\tau_2}=0,\label{neq:64}\\
\hat{A}\hat{P}_{\tau}+\hat{P}_{\tau}\hat{A}^T+e^{\hat{A}\tau_1}\hat{B}\hat{B}^Te^{\hat{A}^T\tau_1}-e^{\hat{A}\tau_2}\hat{B}\hat{B}^Te^{\hat{A}^T\tau_2}=0,\label{neq:65}\\
A^T(\tilde{Y}_\tau+2\tilde{Z}_\tau)+(\tilde{Y}_\tau+2\tilde{Z}_\tau)\hat{A}+e^{A^T\tau_1}C^T\hat{C}e^{\hat{A}\tau_1}-e^{A^T\tau_2}C^T\hat{C}e^{\hat{A}\tau_2}\nonumber\\
+2\sum_{i=1}^{p}\big(e^{A^T\tau_1}M_i\tilde{P}_{\tau}\hat{M_i}e^{\hat{A}\tau_1}- e^{A^T\tau_2}M_i\tilde{P}_{\tau}\hat{M_i}e^{\hat{A}\tau_2}\big)=0,\label{neq:68}\\
\hat{A}^T(\hat{Y}_\tau+2\hat{Z}_{\tau})+(\hat{Y}_\tau+2\hat{Z}_{\tau})\hat{A}+e^{\hat{A}^T\tau_1}\hat{C}^T\hat{C}e^{\hat{A}\tau_1}-e^{\hat{A}^T\tau_2}\hat{C}^T\hat{C}e^{\hat{A}\tau_2}\nonumber\\
+2\sum_{i=1}^{p}\big(e^{\hat{A}^T\tau_1}\hat{M_i}\hat{P}_{\tau}\hat{M_i}e^{\hat{A}\tau_1}-e^{\hat{A}^T\tau_2}\hat{M_i}\hat{P}_{\tau}\hat{M_i}e^{\hat{A}\tau_2}\big)=0.\label{neq:69}
\end{align}
We now present our proposed algorithm named as the ``Time-limited $\mathcal{H}_2$ Near-optimal Iterative Algorithm (TLHNOIA)". This algorithm begins with an arbitrary initial guess of the ROM and iteratively refines it until convergence, which is indicated by minimal change in the ROM's state-space matrices. Steps (\ref{step4}) and (\ref{step5}) calculate the projection matrices in each iteration, while steps (\ref{step6})-(\ref{step10}) employ bi-orthogonal Gram–Schmidt method for ensuring Petrov-Galerkin condition $\hat{W}^T\hat{V}=I$.
\begin{algorithm}
\caption{FLHNOIA}

\textbf{Input:} Full order system: $(A,B,C,M_1,\cdots,M_p)$; Desired time interval: $[\tau_1,\tau_2]$ sec; Initial guess of ROM: $(\hat{A},\hat{B},\hat{C},\hat{M_1},\cdots,\hat{M_p})$; Tolerance: $tol$.

\textbf{Output:} ROM: $(\hat{A},\hat{B},\hat{C},\hat{M_1},\cdots,\hat{M_p})$.

\begin{algorithmic}[1]\label{alg1}
\STATE \textbf{while}\big(relative change in $(\hat{A},\hat{B},\hat{C},\hat{M_1},\cdots,\hat{M_p})$ $>$ $tol$\big)
\STATE Solve equations (\ref{neq:64})-(\ref{neq:69}) to compute $\tilde{P}_{\tau}$, $\hat{P}_{\tau}$, $\tilde{Y}_{\tau}+2\tilde{Z}_{\tau}$, and $\hat{Y}_{\tau}+2\hat{Z}_{\tau}$.\label{step4}
\STATE Set $\hat{V}=\tilde{P}_{\tau}\hat{P}_{\tau}^{-1}$ and $\hat{W}=(\tilde{Y}_{\tau}+2\tilde{Z}_{\tau})(\hat{Y}_{\tau}+2\hat{Z}_{\tau})^{-1}$.\label{step5}
\STATE \textbf{for} $l=1,\ldots,k$ \textbf{do}\label{step6}
\STATE $v=\hat{V}(:,l)$, $v=\prod_{j=1}^{l}\big(I-\hat{V}(:,j)\hat{W}(:,j)^T\big)v$.
\STATE $w=\hat{W}(:,l)$, $w=\prod_{j=1}^{l}\big(I-\hat{W}(:,j)\hat{V}(:,j)^T\big)w$.
\STATE $v=\frac{v}{||v||_2}$, $w=\frac{w}{||w||_2}$, $v=\frac{v}{w^Tv}$, $\hat{V}(:,l)=v$, $\hat{W}(:,l)=w$.
\STATE \textbf{end for}\label{step10}
\STATE Update $\hat{A}=\hat{W}^TA\hat{V}$, $\hat{B}=\hat{W}^TB$, $\hat{C}=C\hat{V}$, $\hat{M_i}=\hat{V}^TM_i\hat{V}$.
\STATE \textbf{end while}
\end{algorithmic}
\end{algorithm}
\begin{remark}
To assess convergence, monitoring the stagnation of the ROM poles is a more reliable indicator than examining state-space realizations. This is because $\mathcal{H}_2$-optimal MOR techniques frequently generate ROMs with different state-space representations but identical transfer functions. The stagnation of ROM poles has been widely adopted as a convergence criterion in $\mathcal{H}_2$-optimal MOR algorithms due to its effectiveness \citep{goyal2019time}.
\end{remark}
\section{Computational Aspects}
This section discusses computational efficiency in implementing TLHNOIA. A key step in TLHNOIA, Step (\ref{step4}), involves calculating the matrix exponential $e^{A\tau}$, which can become computationally intensive for high-order models. To mitigate this, we can utilize Krylov subspace-based methods in \citep{kurschner2018balanced}, to approximate $e^{A\tau}B$, $Ce^{A\tau}$, and $M_ie^{A\tau}$. The most computationally demanding task within each iteration is solving Sylvester equations \ref{neq:64}) and (\ref{neq:68}). Due to the sparsity of state-space matrices in high-order dynamical systems, these Sylvester equations exhibit a \textit{``sparse-dense"} structure, a common feature in $\mathcal{H}_2$-optimal MOR algorithms. Specifically, these equations involve large sparse matrices and smaller dense matrices, i.e.,
\begin{align}
\mathcal{K}\mathcal{L}+\mathcal{L}\mathcal{M}+\mathcal{N}\mathcal{O}&=0,\nonumber
\end{align}wherein the large matrices $\mathcal{K}\in\mathbb{R}^{N\times N}$ and $\mathcal{N}\in\mathbb{R}^{N\times d}$ ($d\ll N$) are sparse, and the small matrices $\mathcal{M}\in\mathbb{R}^{n\times n}$ and $\mathcal{O}\in\mathbb{R}^{d\times n}$ are dense. To efficiently solve such equations, the efficient algorithm proposed in \citep{MPIMD11-11} can be used. The remaining steps of TLHNOIA involve relatively simple matrix computations and small Lyapunov equation solutions, which are computationally inexpensive.
\section{Illustrative Example}
This section presents an illustrative example to verify key properties of TLHNOIA. Consider a sixth-order LQO system defined by the following state-space representation:
\begin{align}
A &= \begin{bmatrix}0 & 0 & 0 & 1 & 0 & 0\\
    0 & 0 & 0 & 0 & 1 & 0\\
    0 & 0 & 0 & 0 & 0 & 1\\
    -5.4545 & 4.5455 & 0 & -0.0545 & 0.0455 & 0\\
    10 & -21 & 11 & 0.1 & -0.21 & 0.11\\
    0 & 5.5 & -6.5 & 0 & 0.055 & -0.065\end{bmatrix},\nonumber\\
B &= \begin{bmatrix}0 & 0 & 0 & 0.0909 & 0.4 & -0.5\end{bmatrix}^T,\nonumber\\
C &= \begin{bmatrix}2 & -2 & 3 & 0 & 0 & 0\end{bmatrix},\nonumber\\
M_1&=diag(0.5, 0.3, 0, 0, 0, 0).\nonumber
\end{align}
The desired time interval for this example is $[0,0.5]$ sec. To initialize TLHNOIA, the following initial guess is used:
\begin{align}
\hat{A}&=\begin{bmatrix}-0.0038&-0.8737&0.0046\\0.8737&-0.0038&0.0053\\0.0054&-0.0060&-0.0353\end{bmatrix},&\hat{B}&=\begin{bmatrix}0.3518&-0.3472&-0.2617\end{bmatrix}^T,\nonumber\\
\hat{C}&=\begin{bmatrix}-0.3454&-0.3405&0.2479\end{bmatrix},&\hat{M}_1&=\begin{bmatrix}0.0113&0.0114&0.0130\\0.0114&0.0116&0.0132\\0.0130&0.0132&0.0271\end{bmatrix}.\nonumber
\end{align} TLHNOIA was stopped when the change in eigenvalues of $\hat{A}$ stagnates, as change in the ROM's state-space realization did not cease. The resulting final ROM is:
\begin{align}
\hat{A}&=\begin{bmatrix}-6.4780&18.5589&5.5099\\-3.1636&5.9676&2.1561\\0.4278&-3.1726&0.8857\end{bmatrix},&\hat{B}&=\begin{bmatrix}0.3533&0.1300&-0.0892\end{bmatrix}^T,\nonumber\\
\hat{C}&=\begin{bmatrix}-0.6333&3.0822&1.9898\end{bmatrix},&\hat{M}_1&=\begin{bmatrix}0.0788&-0.1923&0.0385\\-0.1923&0.4899&-0.0638\\0.0385&-0.0638&0.0637\end{bmatrix}.\nonumber
\end{align}
The numerical results below confirm that this ROM, for all practical purposes, satisfies the optimality conditions (\ref{op2})-(\ref{op4}).
\begin{align}
||-\tilde{P}_{\tau}^TM_i\tilde{P}_{\tau}+\hat{P}_{\tau}\hat{M_i}\hat{P}_{\tau}||_2&=4.3029\times 10^{-9},\nonumber\\
||-(\tilde{Y}_{\tau}+2\tilde{Z}_{\tau})^TB+(\hat{Y}_{\tau}+2\hat{Z}_{\tau})\hat{B}||_2&=1.8855\times 10^{-4},\nonumber\\
||-C\tilde{P}_{\tau}+\hat{C}\hat{P}_{\tau}||_2&=1.5081\times 10^{-5}.\nonumber
\end{align}
Subsequently, a third-order ROM is generated using BT, TLBT, and HOMORA. The same initial ROM is employed to start HOMORA. Figure \ref{fig1} presents the relative output error on a logarithmic scale for the input $u(t) = 0.01cos(2t)$ within the specified $0$ to $0.5$ sec interval. As shown, TLBT and TLHNOIA exhibit superior accuracy.
\begin{figure}[!h]
  \centering
  \includegraphics[width=12cm]{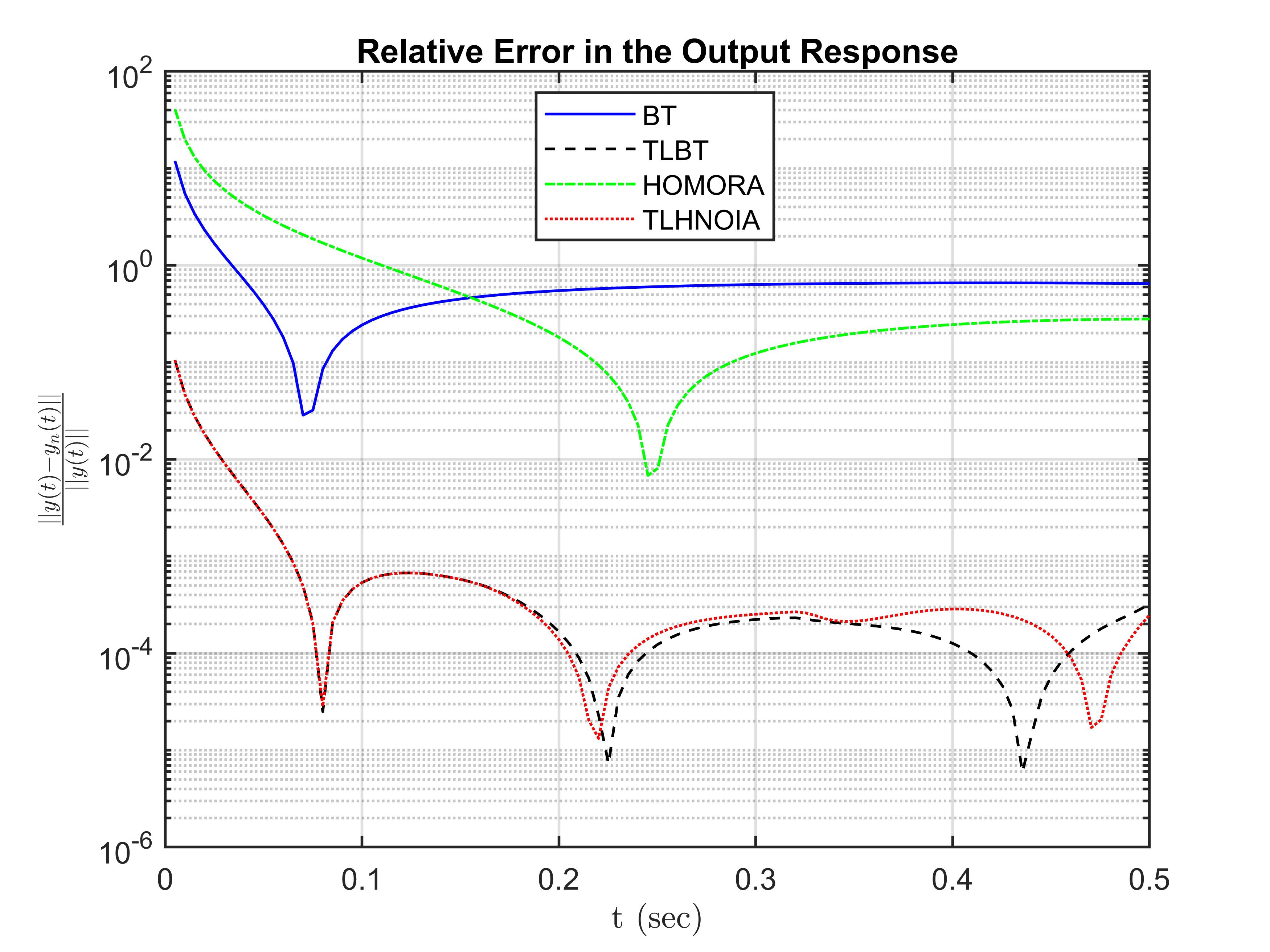}
  \caption{Relative Error in the Output Response within $[0,0.5]$ sec}\label{fig1}
\end{figure}
\section{Conclusion}
This research investigates the $\mathcal{H}_2$-optimal MOR problem within a specified finite time horizon. To quantify output strength within this interval, we introduce the time-limited $\mathcal{H}_2$ norm for LQO systems. We derive necessary conditions for achieving local optima of the squared time-limited $\mathcal{H}_2$ norm of the error. These conditions are compared to those of the standard, unconstrained $\mathcal{H}_2$-optimal MOR problem. We analyze the limitations of the Petrov-Galerkin projection method in satisfying all optimality conditions for the time-limited scenario. Consequently, we propose a Petrov-Galerkin projection algorithm that satisfies three of the four optimality conditions. Numerical experiments validate our theoretical findings and demonstrate the algorithm's effectiveness in achieving high accuracy within the desired time frame.
\section*{Appendix A}
This appendix provides a proof of Theorem \ref{th1}. Throughout the proof, the following trace properties are utilized: the invariance of the trace under matrix transposition and cyclic permutations, and the linearity of the trace operation, i.e.,
\begin{enumerate}
  \item Trace of transpose: $trace(STU)=trace(U^TT^TS^T)$.
  \item Circular permutation in trace: $trace(STU)=trace(UST)=trace(TUS)$.
  \item Trace of addition: $trace(S+T+U)=trace(S)+trace(T)+trace(U)$.
\end{enumerate}
We define a cost function, $J$, as the part of $||E||_{H_{2,\tau}}^2$ that depends on the ROM:
\begin{align}
J=trace(-2B^T\tilde{Q}_{\tau}\hat{B}+\hat{B}^T\hat{Q}_{\tau}\hat{B}).\nonumber
\end{align}
Introducing a small first-order perturbation, $\Delta_{\hat{A}}$ to $\hat{A}$, induces corresponding perturbations, $\Delta_{J}^{\hat{A}}$, $\Delta_{\tilde{Q}_{\tau}}^{\hat{A}}$, and $\Delta_{\hat{Q}_{\tau}}^{\hat{A}}$, in the cost function $J$ and matrices $\tilde{Q}_{\tau}$ and $\hat{Q}_{\tau}$, respectively. The resulting first-order change in $J$ is:
\begin{align}
\Delta_{J}^{\hat{A}}=trace(2B^T\Delta_{\tilde{Q}_{\tau}}^{\hat{A}}\hat{B}+\hat{B}^T\Delta_{\hat{Q}_{\tau}}^{\hat{A}}\hat{B}).\nonumber
\end{align}
Furthermore, based on equations (\ref{eq:30}) and (\ref{eq:31}), we find that $\Delta_{\tilde{Q}_{\tau}}^{\hat{A}}$ and $\Delta_{\hat{Q}_{\tau}}^{\hat{A}}$ are solutions to the following Lyapunov equations:
\begin{align}
&A^T\Delta_{\tilde{Q}_{\tau}}^{\hat{A}}+\Delta_{\tilde{Q}_{\tau}}^{\hat{A}}\hat{A}+\tilde{Q}_{\tau}\Delta_{\hat{A}}-S_{\tau}^TC^T\hat{C}\Delta_{\hat{S}_{\tau}}^{\hat{A}}\nonumber\\
&\hspace*{1.2cm}+\sum_{i=1}^{p}\big(M_i\Delta_{\tilde{P}_{\tau}}^{\hat{A}}\hat{M_i}-S_{\tau}^TM_i\Delta_{\tilde{P}_{\tau}}^{\hat{A}}\hat{M_i}\hat{S}_{\tau}-S_{\tau}^TM_i\tilde{P}_{\tau}\hat{M_i}\Delta_{\hat{S}_{\tau}}^{\hat{A}}\big)=0,\nonumber\\
&\hat{A}^T\Delta_{\hat{Q}_{\tau}}^{\hat{A}}+\Delta_{\hat{Q}_{\tau}}^{\hat{A}}\hat{A}+(\Delta_{\hat{A}})^T\hat{Q}_{\tau}+\hat{Q}_{\tau}\Delta_{\hat{A}}-(\Delta_{\hat{S}_{\tau}}^{\hat{A}})^T\hat{C}^T\hat{C}\hat{S}_{\tau}-\hat{S}_{\tau}^T\hat{C}^T\hat{C}\Delta_{\hat{S}_{\tau}}^{\hat{A}}\nonumber\\
&\hspace*{1.2cm}+\sum_{i=1}^{p}\big(\hat{M}_i\Delta_{\hat{P}_\tau}^{\hat{A}}\hat{M}_i-(\Delta_{\hat{S}_\tau}^{\hat{A}})^T\hat{M}_i\hat{P}_\tau\hat{M}_i\hat{S}_\tau\nonumber\\
&\hspace*{1.2cm}-\hat{S}_{\tau}^T\hat{M}_i\Delta_{\hat{P}_\tau}^{\hat{A}}\hat{M}_i\hat{S}_\tau-\hat{S}_{\tau}^T\hat{M}_i\hat{P}_\tau\hat{M}_i\Delta_{\hat{S}_\tau}^{\hat{A}}\big)=0,\nonumber
\end{align}
where
\begin{align}
\Delta_{\hat{S}_{\tau}}^{\hat{A}}=\mathcal{L}(\hat{A},\Delta_{\hat{A}})+o(||\Delta_{\hat{A}}||);\nonumber
\end{align}cf. \citep{higham2008functions}. As we focus solely on first-order perturbations, the higher-order term $o(||\Delta_{\hat{A}}||)$ is disregarded in subsequent analysis. Now,
\begin{align}
trace\Big(B\hat{B}^T(\Delta_{\tilde{Q}_{\tau}}^{\hat{A}})^T\Big)&=trace\Big((-A\tilde{P}-\tilde{P}\hat{A}^T)(\Delta_{\tilde{Q}_{\tau}}^{\hat{A}})^T\Big)\nonumber\\
&=trace\Big(-A\tilde{P}(\Delta_{\tilde{Q}_{\tau}}^{\hat{A}})^T-\tilde{P}\hat{A}^T(\Delta_{\tilde{Q}_{\tau}}^{\hat{A}})^T\Big)\nonumber\\
&=trace\Big(\tilde{P}^T(-A^T\Delta_{\tilde{Q}_{\tau}}^{\hat{A}}-\hat{A}\Delta_{\tilde{Q}_{\tau}}^{\hat{A}})\Big)\nonumber\\
&=trace\Bigg(\tilde{P}^T\Big(\tilde{Q}_{\tau}\Delta_{\hat{A}}-S_{\tau}^TC^T\hat{C}\Delta_{\hat{S}_{\tau}}^{\hat{A}}\nonumber\\
&\hspace*{1cm}+\sum_{i=1}^{p}\big(M_i\Delta_{\tilde{P}_{\tau}}^{\hat{A}}\hat{M_i}-S_{\tau}^TM_i\Delta_{\tilde{P}_{\tau}}^{\hat{A}}\hat{M_i}\hat{S}_{\tau}\nonumber\\
&\hspace*{1cm}-S_{\tau}^TM_i\tilde{P}_{\tau}\hat{M_i}\Delta_{\hat{S}_{\tau}}^{\hat{A}}\big)\Big)\Bigg)\nonumber\\
&=trace\Big(\tilde{Q}_\tau\tilde{P}(\Delta_{\hat{A}})^T-\hat{C}^TCS_\tau\tilde{P}(\Delta_{\hat{S}_\tau}^{\hat{A}})^T\nonumber\\
&\hspace*{1cm}+\sum_{i=1}^{p}\big(M_i\tilde{P}\hat{M}_i(\Delta_{\tilde{P}_\tau}^{\hat{A}})^T-M_iS_\tau\tilde{P}\hat{S}_{\tau}^T\hat{M}_i(\Delta_{\hat{S}_\tau}^{\hat{A}})^T\nonumber\\
&\hspace*{1cm}-\hat{M}_i\tilde{P}_{\tau}^TM_iS_\tau\tilde{P}(\Delta_{\hat{S}_\tau}^{\hat{A}})^T\big)\Big).\nonumber
\end{align}
Similarly, we find that
\begin{align}
trace(\hat{B}\hat{B}^T\Delta_{\hat{Q}_{\tau}}^{\hat{A}})&=trace\Big(\big(-\hat{A}\hat{P}-\hat{P}\hat{A}^T\big)\Delta_{\hat{Q}_{\tau}}^{\hat{A}}\Big)\nonumber\\
&=trace\Big(\hat{P}\big(-\hat{A}^T\Delta_{\hat{Q}_{\tau}}^{\hat{A}}-\Delta_{\hat{Q}_{\tau}}^{\hat{A}}\hat{A}\big)\Big)\nonumber\\
&=trace\Bigg(\hat{P}\Big((\Delta_{\hat{A}})^T\hat{Q}_{\tau}+\hat{Q}_{\tau}\Delta_{\hat{A}}\nonumber\\
&-(\Delta_{\hat{S}_{\tau}}^{\hat{A}})^T\hat{C}^T\hat{C}\hat{S}_{\tau}-\hat{S}_{\tau}^T\hat{C}^T\hat{C}\Delta_{\hat{S}_{\tau}}^{\hat{A}}\nonumber\\
&+\sum_{i=1}^{p}\big(\hat{M}_i\Delta_{\hat{P}_\tau}^{\hat{A}}\hat{M}_i-(\Delta_{\hat{S}_\tau}^{\hat{A}})^T\hat{M}_i\hat{P}_\tau\hat{M}_i\hat{S}_\tau\nonumber\\
&-\hat{S}_{\tau}^T\hat{M}_i\Delta_{\hat{P}_\tau}^{\hat{A}}\hat{M}_i\hat{S}_\tau-\hat{S}_{\tau}^T\hat{M}_i\hat{P}_\tau\hat{M}_i\Delta_{\hat{S}_\tau}^{\hat{A}}\big)\Big)\Bigg)\nonumber\\
&=trace\Big(2\hat{Q}_\tau\hat{P}(\Delta_{\hat{A}})^T-2\hat{P}\hat{S}_{\tau}^T\hat{C}^T\hat{C}\Delta_{\hat{S}_\tau}^{\hat{A}}\nonumber\\
&+\sum_{i=1}^{p}\big(\hat{M}_i\hat{P}\hat{M}_i\Delta_{\hat{P}_\tau}^{\hat{A}}-\hat{M}_i\hat{S}_\tau\hat{P}\hat{S}_{\tau}^T\hat{M}_i\Delta_{\hat{P}_\tau}^{\hat{A}}\nonumber\\
&-2\hat{P}\hat{S}_{\tau}^T\hat{M}_i\hat{P}_\tau\hat{M}_i\Delta_{\hat{S}_{\tau}}^{\hat{A}}\big)\Big).\nonumber
\end{align}
Further, since
\begin{align}
\tilde{P}_{\tau}=\tilde{P}-S_\tau\tilde{P}\hat{S}_{\tau}^T,\nonumber\\
\hat{P}_{\tau}=\hat{P}-\hat{S}_\tau\hat{P}\hat{S}_{\tau}^T,\nonumber
\end{align}
$\Delta_{J}^{\hat{A}}$ becomes
\begin{align}
\Delta_{J}^{\hat{A}}&=trace\Big(-2(\tilde{Q}_{\tau})^T\tilde{P}(\Delta_{\hat{A}})^T+2\hat{Q}_{\tau}\hat{P}(\Delta_{\hat{A}})^T\nonumber\\
&\hspace*{1cm}+2\tilde{P}^TS_{\tau}^TC^T\hat{C}\Delta_{\hat{S}_\tau}^{\hat{A}}-2\hat{P}\hat{S}_{\tau}^T\hat{C}^T\hat{C}\Delta_{\hat{S}_\tau}^{\hat{A}}\nonumber\\
&\hspace*{1cm}+\sum_{i=1}^{p}\big(-2M_i\tilde{P}_\tau\hat{M}_i(\Delta_{\tilde{P}_\tau}^{\hat{A}})^T+2\tilde{P}^TS_{\tau}^TM_i\tilde{P}_\tau\hat{M}_i\Delta_{\hat{S}_\tau}^{\hat{A}}\nonumber\\
&\hspace*{1cm}+\hat{M}_i\hat{P}_\tau\hat{M}_i\Delta_{\hat{P}_\tau}^{\hat{A}}-2\hat{P}\hat{S}_{\tau}^T\hat{M}_i\hat{P}_\tau\hat{M}_i\Delta_{\hat{S}_{\tau}}^{\hat{A}}\big)\Big);\nonumber
\end{align}cf. \citep{zulfiqar2020time}.

Introducing a small first-order perturbation, $\Delta_{\hat{A}}$ to $\hat{A}$, induces corresponding perturbations, $\Delta_{\tilde{P}_{\tau}}^{\hat{A}}$ and $\Delta_{\hat{P}_{\tau}}^{\hat{A}}$. Based on equations (\ref{eq:24}) and (\ref{eq:25}), we find that $\Delta_{\tilde{P}_{\tau}}^{\hat{A}}$ and $\Delta_{\hat{P}_{\tau}}^{\hat{A}}$ are solutions to the following Lyapunov equations:
\begin{align}
&A\Delta_{\tilde{P}_{\tau}}^{\hat{A}}+\Delta_{\tilde{P}_{\tau}}^{\hat{A}}\hat{A}^T+\tilde{P}_{\tau}(\Delta_{\hat{A}})^T+S_{\tau}B\hat{B}^T(\Delta_{\hat{S}_{\tau}}^{\hat{A}})^T=0,\nonumber\\
&\hat{A}\Delta_{\hat{P}_{\tau}}^{\hat{A}}+\Delta_{\hat{P}_{\tau}}^{\hat{A}}\hat{A}^T+\Delta_{\hat{A}}\hat{P}_{\tau}+\hat{P}_{\tau}(\Delta_{\hat{A}})^T\nonumber\\
&\hspace*{1cm}-\Delta_{\hat{S}_{\tau}}^{\hat{A}}\hat{B}\hat{B}^T\hat{S}_{\tau}^T-\hat{S}_{\tau}\hat{B}\hat{B}^T(\Delta_{\hat{S}_{\tau}}^{\hat{A}})^T=0.\nonumber
\end{align}
Observe that
\begin{align}
trace\Big(\sum_{i=1}^{p}M_i\tilde{P}_{\tau}\hat{M_i}(\Delta_{\tilde{P}_{\tau}}^{\hat{A}})^T\Big)&=trace\Big(\big(-A^T\bar{Z}_{\tau}-\bar{Z}_{\tau}\hat{A}\big)(\Delta_{\tilde{P}_{\tau}}^{\hat{A}})^T\Big)\nonumber\\
&=trace\Big(\big(-A\Delta_{\tilde{P}_{\tau}}^{\hat{A}}-\Delta_{\tilde{P}_{\tau}}^{\hat{A}}\hat{A}^T\big)\bar{Z}_{\tau}^T\Big)\nonumber\\
&=trace\big(\bar{Z}_{\tau}^T\tilde{P}_\tau(\Delta_{\hat{A}})^T-\hat{B}B^TS_{\tau}^T\bar{Z}_{\tau}\Delta_{\hat{S}_{\tau}}^{\hat{A}}\big),\nonumber
\end{align}
and
\begin{align}
trace\Big(\sum_{i=1}^{p}\hat{M_i}\hat{P}_{\tau}\hat{M_i}\Delta_{\hat{P}_{\tau}}^{\hat{A}}\Big)&=trace\Big(\big(-\hat{A}^T\bar{Z}_{n,\tau}-\bar{Z}_{n,\tau}\hat{A}\big)\Delta_{\hat{P}_{\tau}}^{\hat{A}}\Big)\nonumber\\
&=trace\Big(\big(-\hat{A}\Delta_{\hat{P}_{\tau}}^{\hat{A}}-\Delta_{\hat{P}_{\tau}}^{\hat{A}}\hat{A}^T\big)\bar{Z}_{n,\tau}\Big)\nonumber\\
&=2trace\big(\bar{Z}_{n,\tau}\hat{P}_{\tau}(\Delta_{\hat{A}})^T-\hat{B}\hat{B}^T\hat{S}_{\tau}^T\bar{Z}_{n,\tau}\Delta_{\hat{S}_\tau}^{\hat{A}}\big).\nonumber
\end{align}
Consequently, $\Delta_{J}^{\hat{A}}$ becomes
\begin{align}
\Delta_{J}^{\hat{A}}&=trace\Big(-2(\tilde{Q}_{\tau})^T\tilde{P}(\Delta_{\hat{A}})^T+2\hat{Q}_{\tau}\hat{P}(\Delta_{\hat{A}})^T\nonumber\\
&\hspace*{0.5cm}-2\bar{Z}_{\tau}^T\tilde{P}_{\tau}(\Delta_{\hat{A}})^T+2\bar{Z}_{n,\tau}\hat{P}_{\tau}(\Delta_{\hat{A}})^T\nonumber\\
&\hspace*{0.5cm}+2\hat{B}B^TS_{\tau}^T\bar{Z}_{\tau}\Delta_{\hat{S}_{\tau}}^{\hat{A}}-2\hat{B}\hat{B}^T\hat{S}_{\tau}^T\bar{Z}_{n,\tau}\Delta_{\hat{S}_{\tau}}^{\hat{A}}\nonumber\\
&\hspace*{0.5cm}+2\tilde{P}^TS_{\tau}^TC^T\hat{C}\Delta_{\hat{S}_{\tau}}^{\hat{A}}-2\hat{P}\hat{S}_{\tau}^T\hat{C}\hat{C}^T\Delta_{\hat{S}_{\tau}}^{\hat{A}}\nonumber\\
&\hspace*{0.5cm}+\sum_{i=1}^{p}\big(2\tilde{P}^TS_{\tau}^TM_i\tilde{P}_{\tau}\hat{M_i}\Delta_{\hat{S}_{\tau}}^{\hat{A}}-2\hat{P}\hat{S}_{\tau}^T\hat{M_i}\hat{P}_{\tau}\hat{M_i}\Delta_{\hat{S}_{\tau}}^{\hat{A}}\big)\Big)\nonumber\\
&=trace\Big(-2(\tilde{Q}_{\tau})^T\tilde{P}(\Delta_{\hat{A}})^T+2\hat{Q}_{\tau}\hat{P}(\Delta_{\hat{A}})^T\nonumber\\
&\hspace*{0.5cm}-2\bar{Z}_{\tau}^T\tilde{P}_{\tau}(\Delta_{\hat{A}})^T+2\bar{Z}_{n,\tau}\hat{P}_{\tau}(\Delta_{\hat{A}})^T+2V\Delta_{\hat{S}_{\tau}}^{\hat{A}}\Big).\nonumber
\end{align}
Recall that $\Delta_{\hat{S}_{\tau}}^{\hat{A}}=\mathcal{L}(\hat{A},\Delta_{\hat{A}})$. Interchanging the trace and integral operators yields:
\begin{align}
trace(V\Delta_{\hat{S}_{\tau}}^{\hat{A}})=trace(W\Delta_{\hat{A}});\nonumber
\end{align}cf. \citep{higham2008functions,petersen2008matrix}. Resultantly,
\begin{align}
\Delta_{J}^{\hat{A}}=2trace\Big(\big(-\tilde{Q}_{\tau}^T\tilde{P}+\hat{Q}_{\tau}\hat{P}-\bar{Z}_{\tau}^T\tilde{P}_{\tau}+\bar{Z}_{n,\tau}\hat{P}_{\tau}+W^T\big)(\Delta_{\hat{A}})^T\Big)\nonumber
\end{align}
Consequently, the gradient of $J$ with respect to $\hat{A}$ is:
\begin{align}
\nabla_{J}^{\hat{A}}=2\big(-\tilde{Q}_{\tau}^T\tilde{P}+\hat{Q}_{\tau}\hat{P}-\bar{Z}_{\tau}^T\tilde{P}_{\tau}+\bar{Z}_{n,\tau}\hat{P}_{\tau}+W^T\big);\nonumber
\end{align}cf. \citep{reiter2024h2}.
It is evident that
\begin{align}
-\tilde{Q}_{\tau}^T\tilde{P}+\hat{Q}_{\tau}\hat{P}-\bar{Z}_{\tau}^T\tilde{P}_{\tau}+\bar{Z}_{n,\tau}\hat{P}_{\tau}+W^T=0\label{nst79}
\end{align} is a necessary condition for the local optimum of $||E||_{H_{2,\tau}}^2$. By substituting equations (\ref{nst46})-(\ref{nst49}) into (\ref{nst79}), we obtain:
\begin{align}
-\tilde{Q}_{\tau}^T\tilde{P}_{\tau}+\hat{Q}_{\tau}\hat{P}_{\tau}-\tilde{Z}_{\tau}^T\tilde{P}_{\tau}+\hat{Z}_{\tau}\hat{P}_{\tau}+L_\tau=0.\nonumber
\end{align} Given that $\tilde{Q}_{\tau}=\tilde{Y}_{\tau}+\tilde{Z}_{\tau}$ and $\hat{Q}_{\tau}=\hat{Y}_{\tau}+\hat{Z}_{\tau}$, the equation simplifies to:
\begin{align}
-(\tilde{Y}_{\tau}+2\tilde{Z}_{\tau})^T\tilde{P}_{\tau}+(\hat{Y}_{\tau}+2\hat{Z}_{\tau})\hat{P}_{\tau}+L_\tau&=0.\nonumber
\end{align}

A small first-order change to the matrix $\hat{M_i}$, denoted as $\Delta_{\hat{M_i}}$, induces corresponding changes in other variables: $J$ becomes $J+\Delta_{J}^{\hat{M_i}}$, $\tilde{Q}_{\tau}$ becomes $\tilde{Q}_{\tau}+\Delta_{\tilde{Q}_{\tau}}^{\hat{M_i}}$, and $\hat{Q}_{\tau}$ becomes $\hat{Q}_{\tau}+\Delta_{\hat{Q}_{\tau}}^{\hat{M_i}}$. The resulting first-order perturbation in $J$, denoted $\Delta_{J}^{\hat{M_i}}$, can be expressed as
\begin{align}
\Delta_{J}^{\hat{M_i}}=trace(-2B^T\Delta_{\tilde{Q}_{\tau}}^{\hat{M_i}}\hat{B}+\hat{B}^T\Delta_{\hat{Q}_{\tau}}^{\hat{M_i}}\hat{B}).\nonumber
\end{align}
Furthermore, based on equations (\ref{eq:30}) and (\ref{eq:31}), $\Delta_{\tilde{Q}_{\tau}}^{\hat{M_i}}$ and $\Delta_{\hat{Q}_{\tau}}^{\hat{M_i}}$ are solutions to the following Lyapunov equations:
\begin{align}
&A^T\Delta_{\tilde{Q}_{\tau}}^{\hat{M_i}}+\Delta_{\tilde{Q}_{\tau}}^{\hat{M_i}}\hat{A}+M_i\tilde{P}_\tau\Delta_{\hat{M}_i}-S_{\tau}^TM_i\tilde{P}_\tau\Delta_{\hat{M}_i}\hat{S}_{\tau}=0,\nonumber\\
&\hat{A}^T\Delta_{\hat{Q}_{\tau}}^{\hat{M_i}}+\Delta_{\hat{Q}_{\tau}}^{\hat{M_i}}\hat{A}+\Delta_{\hat{M}_i}\hat{P}_\tau\hat{M}_i+\hat{M}_i\hat{P}_\tau\Delta_{\hat{M}_i}\nonumber\\
&\hspace*{1.15cm}-\hat{S}_{\tau}^T\Delta_{\hat{M}_i}\hat{P}_\tau\hat{M}_i\hat{S}_\tau-\hat{S}_{\tau}^T\hat{M}_i\hat{P}_\tau\Delta_{\hat{M}_i}\hat{S}_\tau=0.\nonumber
\end{align}
Observe that
\begin{align}
trace(B^T\Delta_{\tilde{Q}_{\tau}}^{\hat{M_i}}\hat{B})&=trace\big(B\hat{B}^T(\Delta_{\tilde{Q}_{\tau}}^{\hat{M_i}})^T\big)\nonumber\\
&=trace\Big(\big(-A\tilde{P}-\tilde{P}\hat{A}^T\big)(\Delta_{\tilde{Q}_{\tau}}^{\hat{M_i}})^T\Big)\nonumber\\
&=trace\Big(\big(-A^T\Delta_{\tilde{Q}_{\tau}}^{\hat{M_i}}-\Delta_{\tilde{Q}_{\tau}}^{\hat{M_i}}\hat{A}\big)\tilde{P}^T\Big)\nonumber\\
&=trace\Big(\big(M_i\tilde{P}_\tau\Delta_{\hat{M}_i}-S_{\tau}^TM_i\tilde{P}_\tau\Delta_{\hat{M}_i}\hat{S}_{\tau}\big)\tilde{P}^T\Big)\nonumber\\
&=trace(\tilde{P}_{\tau}M_i\tilde{P}_\tau(\Delta_{\hat{M}_i})^T).\nonumber
\end{align}
Additionally,
\begin{align}
trace(\hat{B}^T\Delta_{\hat{Q}_{\tau}}^{\hat{M_i}}\hat{B})&=trace\big(\hat{B}\hat{B}^T\Delta_{\hat{Q}_{\tau}}^{\hat{M_i}}\big)\nonumber\\
&=trace\Big(\big(-\hat{A}\hat{P}-\hat{P}\hat{A}^T\big)\Delta_{\hat{Q}_{\tau}}^{\hat{M_i}}\Big)\nonumber\\
&=trace\Big(\big(-\hat{A}^T\Delta_{\hat{Q}_{\tau}}^{\hat{M_i}}-\Delta_{\hat{Q}_{\tau}}^{\hat{M_i}}\hat{A}\big)\hat{P}\Big)\nonumber\\
&=trace\Big(\big(\Delta_{\hat{M}_i}\hat{P}_\tau\hat{M}_i+\hat{M}_i\hat{P}_\tau\Delta_{\hat{M}_i}\nonumber\\
&-\hat{S}_{\tau}^T\Delta_{\hat{M}_i}\hat{P}_\tau\hat{M}_i\hat{S}_\tau-\hat{S}_{\tau}^T\hat{M}_i\hat{P}_\tau\Delta_{\hat{M}_i}\hat{S}_\tau\big)\hat{P}\Big)\nonumber\\
&=trace\big(2\hat{P}_{\tau}\hat{M_i}\hat{P}_{\tau}\big).\nonumber
\end{align}
Consequently, $\Delta_{J}^{\hat{M_i}}$ can be expressed as:
\begin{align}
\Delta_{J}^{\hat{M_i}}=2trace\big((-\tilde{P}_{\tau}^TM_i\tilde{P}_{\tau}+\hat{P}_{\tau}\hat{M_i}\hat{P}_{\tau})(\Delta_{\hat{M_i}})^T\big).\nonumber
\end{align}
From this, the gradient of $J$ with respect to $\hat{M_i}$ is:
\begin{align}
\nabla_{J}^{\hat{M_i}}=2(-\tilde{P}_{\tau}^TM_i\tilde{P}_{\tau}+\hat{P}_{\tau}\hat{M_i}\hat{P}_{\tau}).\nonumber
\end{align}
Therefore, a necessary condition for a local minimum of $||E||_{\mathcal{H}_{2,\tau}}^2$ is:
\begin{align}
-\tilde{P}_{\tau}^TM_i\tilde{P}_{\tau}+\hat{P}_{\tau}\hat{M_i}\hat{P}_{\tau}=0.\nonumber
\end{align}

A small perturbation to the matrix $\hat{B}$, denoted as $\Delta_{\hat{B}}$, results in corresponding changes to other variables: $J$ becomes $J+\Delta_{J}^{\hat{B}}$, $\tilde{P}_{\tau}$ becomes $\tilde{P}_{\tau}+\Delta_{\tilde{P}_{\tau}}^{\hat{B}}$, $\hat{P}_{\tau}$ becomes $\hat{P}_{\tau}+\Delta_{\hat{P}_{\tau}}^{\hat{B}}$, $\tilde{Q}_{\tau}$ becomes $\tilde{Q}_{\tau}+\Delta_{\tilde{Q}_{\tau}}^{\hat{B}}$, and $\hat{Q}_{\tau}$ becomes $\hat{Q}_{\tau}+\Delta_{\hat{Q}_{\tau}}^{\hat{B}}$. The resulting first-order change in $J$, represented by $\Delta_{J}^{\hat{B}}$, can be expressed as:
\begin{align}
\Delta_{J}^{\hat{B}}=trace\big(-2\tilde{Q}_{\tau}^TB(\Delta_{\hat{B}})^T+2\hat{Q}_{\tau}\hat{B}(\Delta_{\hat{B}})^T-2B\hat{B}(\Delta_{\tilde{Q}_{\tau}}^{\hat{B}})^T+\hat{B}\hat{B}^T\Delta_{\hat{Q}_{\tau}}^{\hat{B}}\big).\nonumber
\end{align}
Based on equations (\ref{eq:24}) through (\ref{eq:31}), the variables $\Delta_{\tilde{P}_{\tau}}^{\hat{B}}$, $\Delta_{\hat{P}_{\tau}}^{\hat{B}}$, $\Delta_{\tilde{Q}_{\tau}}^{\hat{B}}$, and $\Delta_{\hat{Q}_{\tau}}^{\hat{B}}$ satisfy the following equations:
\begin{align}
A\Delta_{\tilde{P}_{\tau}}^{\hat{B}}+\Delta_{\tilde{P}_{\tau}}^{\hat{B}}\hat{A}^T+ B(\Delta_{\hat{B}})^T-S_{\tau}B(\Delta_{\hat{B}})^T\hat{S}_{\tau}^T=0,\nonumber\\
\hat{A}\Delta_{\hat{P}_{\tau}}^{\hat{B}}+\Delta_{\hat{P}_{\tau}}^{\hat{B}}\hat{A}^T+\Delta_{\hat{B}}\hat{B}^T+\hat{B}(\Delta_{\hat{B}})^T-\hat{S}_{\tau}\Delta_{\hat{B}}\hat{B}^T\hat{S}_{\tau}^T-\hat{S}_\tau\hat{B}(\Delta_{\hat{B}})^T\hat{S}_{\tau}^T=0\nonumber\\
A^T\Delta_{\tilde{Q}_{\tau}}^{\hat{B}}+\Delta_{\tilde{Q}_{\tau}}^{\hat{B}}\hat{A}+\sum_{i=1}^{p}\big(M_i\Delta_{\tilde{P}_{\tau}}^{\hat{B}}\hat{M_i}-S_{\tau}^TM_i\Delta_{\tilde{P}_{\tau}}^{\hat{B}}\hat{M_i}\hat{S}_{\tau}\big)=0,\nonumber\\
\hat{A}^T\Delta_{\hat{Q}_{\tau}}^{\hat{B}}+\Delta_{\hat{Q}_{\tau}}^{\hat{B}}\hat{A}+\sum_{i=1}^{p}\big(\hat{M_i}\Delta_{\hat{P}_{\tau}}^{\hat{B}}\hat{M_i}-\hat{S}_{\tau}^T\hat{M_i}\Delta_{\hat{P}_{\tau}}^{\hat{B}}\hat{M_i}\hat{S}_{\tau}\big)=0.\nonumber
\end{align}
It can be observed that
\begin{align}
trace\big(B\hat{B}^T(\Delta_{\tilde{Q}_{\tau}}^{\hat{B}})^T\big)&=trace\Big(\big(-A\tilde{P}-\tilde{P}\hat{A}^T\big)(\Delta_{\tilde{Q}_{\tau}}^{\hat{B}})^T\Big)\nonumber\\
&=trace\Big(\big(-A^T\Delta_{\tilde{Q}_{\tau}}^{\hat{B}}-\Delta_{\tilde{Q}_{\tau}}^{\hat{B}}\hat{A}\big)\tilde{P}^T\Big)\nonumber\\
&=trace\Bigg(\Big(\sum_{i=1}^{p}\big(M_i\Delta_{\tilde{P}_{\tau}}^{\hat{B}}\hat{M_i}-S_{\tau}^TM_i\Delta_{\tilde{P}_{\tau}}^{\hat{B}}\hat{M_i}\hat{S}_{\tau}\big)\Big)\tilde{P}^T\Bigg)\nonumber\\
&=trace\Big(\sum_{i=1}^{p}M_i\tilde{P}_{\tau}\hat{M_i}(\Delta_{\tilde{P}_{\tau}}^{\hat{B}})^T\Big).\nonumber
\end{align}
Similarly,
\begin{align}
trace\big(\hat{B}\hat{B}^T\Delta_{\hat{Q}_{\tau}}^{\hat{B}}\big)&=trace\Big(\big(-\hat{A}\hat{P}-\hat{P}\hat{A}^T\big)\Delta_{\hat{Q}_{\tau}}^{\hat{B}}\Big)\nonumber\\
&=trace\Big(\big(-\hat{A}^T\Delta_{\hat{Q}_{\tau}}^{\hat{B}}-\Delta_{\hat{Q}_{\tau}}^{\hat{B}}\hat{A}\big)\hat{P}\Big)\nonumber\\
&=trace\Bigg(\Big(\sum_{i=1}^{p}\big(\hat{M_i}\Delta_{\hat{P}_{\tau}}^{\hat{B}}\hat{M_i}-\hat{S}_{\tau}^T\hat{M_i}\Delta_{\hat{P}_{\tau}}^{\hat{B}}\hat{M_i}\hat{S}_{\tau}\big)\Big)\hat{P}\Bigg)\nonumber\\
&=trace\Big(\sum_{i=1}^{p}\hat{M_i}\hat{P}_{\tau}\hat{M_i}\Delta_{\hat{P}_{\tau}}^{\hat{B}}\Big).\nonumber
\end{align}
Therefore, $\Delta_{J}^{\hat{B}}$ can be expressed as:
\begin{align}
\Delta_{J}^{\hat{B}}&=trace\Big(-2\tilde{Q}_{\tau}^TB(\Delta_{\hat{B}})^T+2\hat{Q}_{\tau}\hat{B}(\Delta_{\hat{B}})^T\nonumber\\
&\hspace*{1.5cm}-2\sum_{i=1}^{p}M_i\tilde{P}_{\tau}\hat{M_i}(\Delta_{\tilde{P}_{\tau}}^{\hat{B}})^T+\sum_{i=1}^{p}\hat{M_i}\hat{P}_{\tau}\hat{M_i}\Delta_{\hat{P}_{\tau}}^{\hat{B}}\Big).\nonumber
\end{align}
It can be shown that
\begin{align}
trace\Big(\sum_{i=1}^{p}M_i\tilde{P}_{\tau}\hat{M_i}(\Delta_{\tilde{P}_{\tau}}^{\hat{B}})^T\Big)&=trace\Big(\big(-A^T\bar{Z}_{\tau}-\bar{Z}_{\tau}\hat{A}\big)(\Delta_{\tilde{P}_{\tau}}^{\hat{B}})^T\Big)\nonumber\\
&=trace\Big(\big(-A\Delta_{\tilde{P}_{\tau}}^{\hat{B}}-\Delta_{\tilde{P}_{\tau}}^{\hat{B}}\hat{A}^T\big)\bar{Z}_{\tau}^T\Big)\nonumber\\
&=trace\Big(\big(B(\Delta_{\hat{B}})^T-S_{\tau}B(\Delta_{\hat{B}})^T\hat{S}_{\tau}^T\big)\bar{Z}_{\tau}^T\Big)\nonumber\\
&=trace\big(\tilde{Z}_{\tau}^TB(\Delta_{\hat{B}})^T\big),\nonumber
\end{align}
and
\begin{align}
trace\Big(\sum_{i=1}^{p}\hat{M_i}\hat{P}_{\tau}\hat{M_i}\Delta_{\hat{P}_{\tau}}^{\hat{B}}\Big)&=trace\Big(\big(-\hat{A}^T\bar{Z}_{n,\tau}-\bar{Z}_{n,\tau}\hat{A}\big)\Delta_{\hat{P}_{\tau}}^{\hat{B}}\Big)\nonumber\\
&=trace\Big(\big(-\hat{A}\Delta_{\hat{P}_{\tau}}^{\hat{B}}-\Delta_{\hat{P}_{\tau}}^{\hat{B}}\hat{A}^T\big)\bar{Z}_{n,\tau}\Big)\nonumber\\
&=trace\Bigg(\Big(\Delta_{\hat{B}}\hat{B}^T+\hat{B}(\Delta_{\hat{B}})^T-\hat{S}_{\tau}\Delta_{\hat{B}}\hat{B}^T\hat{S}_{\tau}^T\nonumber\\
&\hspace*{1.5cm}-\hat{S}_\tau\hat{B}(\Delta_{\hat{B}})^T\hat{S}_{\tau}^T\Big)\bar{Z}_{n,\tau}\Bigg)\nonumber\\
&=2trace\big(\hat{Z}_{\tau}^T\hat{B}(\Delta_{\hat{B}})^T\big).\nonumber
\end{align}
As a result, $\Delta_{J}^{\hat{B}}$ simplifies to:
\begin{align}
\Delta_{J}^{\hat{B}}&=trace\big(-2\tilde{Q}_{\tau}^TB(\Delta_{\hat{B}})^T+2\hat{Q}_{\tau}\hat{B}(\Delta_{\hat{B}})^T\nonumber\\
&\hspace*{1.5cm}-2\tilde{Z}_{\tau}^TB(\Delta_{\hat{B}})^T+2\hat{Z}_{\tau}\hat{B}(\Delta_{\hat{B}})^T\big).\nonumber
\end{align}
From this, the gradient of $J$ with respect to $\hat{B}$ is:
\begin{align}
\nabla_{J}^{\hat{B}}=2(-\tilde{Q}_{\tau}^TB+\hat{Q}_{\tau}\hat{B}-\tilde{Z}_{\tau}^TB+\hat{Z}_{\tau}\hat{B}).\nonumber
\end{align}
Therefore, a necessary condition for a local minimum of $||E||_{\mathcal{H}_{2,\tau}}^2$ is:
\begin{align}
-\tilde{Q}_{\tau}^TB+\hat{Q}_{\tau}\hat{B}-\tilde{Z}_{\tau}^TB+\hat{Z}_{\tau}\hat{B}=-(\tilde{Y}_{\tau}+2\tilde{Z}_{\tau})^TB+(\hat{Y}_{\tau}+2\hat{Z}_{\tau})\hat{B}=0.\nonumber
\end{align}

We begin by rewriting the cost function $J$ in a slightly different form:
\begin{align}
J&=trace(-2B^T(\tilde{Y}_{\tau}+\tilde{Z}_{\tau})\hat{B}+\hat{B}^T(\hat{Y}_{\tau}+\hat{Z}_{\tau})\hat{B})\nonumber\\
&=trace(-2B^T\tilde{Y}_{\tau}\hat{B}-2B^T\tilde{Z}_{\tau}\hat{B}+\hat{B}^T\hat{Y}_{\tau}\hat{B}+\hat{B}^T\hat{Z}_{\tau}\hat{B}).\nonumber
\end{align}
Observing that
\begin{align}
trace(-2B^T\tilde{Y}_{\tau}\hat{B}+\hat{B}^T\hat{Y}_{\tau}\hat{B})=trace(2C\tilde{P}_{\tau}\hat{C}^T+\hat{C}\hat{P}_{\tau}\hat{C}^T),\nonumber
\end{align}
we can rewrite $J$ as:
\begin{align}
J=trace(2C\tilde{P}_{\tau}\hat{C}^T+\hat{C}\hat{P}_{\tau}\hat{C}^T-2B^T\tilde{Z}_{\tau}\hat{B}+\hat{B}^T\hat{Z}_{\tau}\hat{B});\nonumber
\end{align}cf. \citep{goyal2019time}. Introducing a small perturbation to the matrix $\hat{C}$, denoted as $\Delta_{\hat{C}}$, leads to a corresponding change in $J$, expressed as $J+\Delta_{J}^{\hat{C}}$. The resulting first-order change in $J$, represented by $\Delta_{J}^{\hat{C}}$, is given by:
\begin{align}
\Delta_{J}^{\hat{C}}=trace(2C\tilde{P}_{\tau}(\Delta_{\hat{C}})^T+2\hat{C}\hat{P}_{\tau}(\Delta_{\hat{C}})^T)\nonumber
\end{align}
Consequently, the gradient of $J$ with respect to $\hat{C}$ is:
\begin{align}
\nabla_{J}^{\hat{C}}=2(C\tilde{P}_{\tau}+\hat{C}\hat{P}_{\tau}).\nonumber
\end{align}
Therefore, a necessary condition for a local minimum of $||E||_{\mathcal{H}_{2,\tau}}^2$ is:
\begin{align}
C\tilde{P}_{\tau}+\hat{C}\hat{P}_{\tau}=0.\nonumber
\end{align}
This concludes the proof.
\section*{Appendix B}
Pre-multiplying equation (\ref{eq:24}) by $\hat{W}^T$ yields:
\begin{align}
\hat{W}^T\big(A\tilde{P}_{\tau}+\tilde{P}_{\tau}\hat{A}^T+ B\hat{B}^T-S_{\tau}B\hat{B}^T\hat{S}_{\tau}^T\big)&=0\nonumber\\
\hat{A}+\hat{A}^T+\hat{B}\hat{B}^T-\hat{S}_{\tau}\hat{B}\hat{B}^T\hat{S}_{\tau}&=0.\nonumber
\end{align} Given the uniqueness of equation (\ref{eq:25}), we conclude that $\hat{P}_{\tau}=I$.

It can readily be noted that $(\tilde{Y}_\tau+2\tilde{Z}_\tau)^T$ and $\hat{Y}_\tau+2\hat{Z}_\tau$ satisfy the following equations:
\begin{align}
&\hat{A}^T(\tilde{Y}_\tau+2\tilde{Z}_\tau)^T+(\tilde{Y}_\tau+2\tilde{Z}_\tau)^TA+\hat{C}^TC-\hat{S}_{\tau}^T\hat{C}^TCS_{\tau}\nonumber\\
&\hspace*{2.4cm}+\sum_{i=1}^{p}2\hat{M_i}\tilde{P}_{\tau}^TM_i-\sum_{i=1}^{p}2\hat{S}_{\tau}^T\hat{M_i}\tilde{P}_{\tau}^TM_iS_{\tau}=0,\label{lneq:44}\\
&\hat{A}^T(\hat{Y}_\tau+2\hat{Z}_\tau)+(\hat{Y}_\tau+2\hat{Z}_\tau)\hat{A}+\hat{C}^T\hat{C}-\hat{S}_{\tau}^T\hat{C}^T\hat{C}\hat{S}_{\tau}\nonumber\\
&\hspace*{2.4cm}+\sum_{i=1}^{p}2\hat{M_i}\hat{P}_{\tau}^T\hat{M}_i-\sum_{i=1}^{p}2\hat{S}_{\tau}^T\hat{M_i}\hat{P}_{\tau}^T\hat{M}_i\hat{S}_{\tau}=0.\label{lneq:45}
\end{align}
By post-multiplying (\ref{lneq:44}) by $\hat{V}$ results in:
\begin{align}
&\Big(\hat{A}^T(\tilde{Y}_\tau+2\tilde{Z}_\tau)^T+(\tilde{Y}_\tau+2\tilde{Z}_\tau)^TA+\hat{C}^TC-\hat{S}_{\tau}^T\hat{C}^TCS_{\tau}\nonumber\\
&\hspace*{1.5cm}+\sum_{i=1}^{p}2\hat{M_i}\tilde{P}_{\tau}^TM_i-\sum_{i=1}^{p}2\hat{S}_{\tau}^T\hat{M_i}\tilde{P}_{\tau}^TM_iS_{\tau}\Big)\hat{V}\nonumber\\
&=\hat{A}^T+\hat{A}+\hat{C}^T\hat{C}-\hat{S}_{\tau}^T\hat{C}^T\hat{C}\hat{S}_{\tau}\nonumber\\
&\hspace*{1.5cm}+\sum_{i=1}^{p}2\hat{M_i}\hat{M_i}-\sum_{i=1}^{p}2\hat{S}_{\tau}^T\hat{M_i}\hat{M_i}\hat{S}_{\tau}=0.\nonumber
\end{align}Due to the uniqueness of equations (\ref{eq:25}) and (\ref{lneq:45}),we find that $\hat{P}_{\tau}=I$ and $\hat{Y}_\tau+2\hat{Z}_\tau=I$. Consequently, the optimality conditions (\ref{op2})-(\ref{op4}) are satisfied.

\end{document}